\documentclass[letterpaper, 11 pt, conference]{ieeeconf}  
\usepackage[utf8]{inputenc}
 \usepackage{amsfonts}	
\usepackage{float}
\usepackage{amssymb}
\usepackage{mathtools}	
\usepackage{mathabx}
\usepackage{xcolor} 
\usepackage{graphics}
\usepackage{url}
\usepackage{booktabs}

\newcommand{\blue}[1]{{\textcolor{black}{#1}}}

\DeclareMathOperator{\Dom}{Dom}	

\newtheorem{remark}{Remark}
\newtheorem{theorem}{Theorem}
\newtheorem{assumption}{Assumption}

\newtheorem{lemma}{Lemma}

\newcommand{\brackets}[1]{\left[#1\right]}
\newcommand{\braces}[1]{\left(#1\right)}
\newcommand{\angles}[1]{\langle#1\rangle}

\newcommand{\vect}[1]{\boldsymbol{#1}}
\newcommand{\setof}[1]{\left\{#1\right\}}
\newcommand{\norm}[1]{\|#1\|}

\DeclareMathOperator{\col}{col}

\newcommand{\qedwhite}{\hfill \ensuremath{\Box}}

\usepackage[nolist]{acronym}
\usepackage{hyphenat}

\allowdisplaybreaks

\title{On modelling and stabilizability of current-controlled piezoelectric
material}

\author{Matthijs C. de Jong and Jacquelien M. A. Scherpen 
		\thanks{ M.C. de Jong and J.M.A. Scherpen are with Jan C. Wilems Center for Systems and Control, ENTEG, Faculty of Science and Engineering, University of Groningen, Nijenborgh 4, 9747 AG Groningen, the Netherlands (email: \{  matthijs.de.jong, j.m.a.scherpen\}@rug.nl).}
	}

\begin{document}

\maketitle
\pagenumbering{arabic}
\begin{abstract}
      This paper presents a new modelling approach to fully dynamic electromagnetic current-controlled piezoelectric composite models that require a combined Lagrangian. To model the mechanical domains, we consider two different beam theories, i.e. the Euler-Bernoulli and Timoshenko beam theories. We show that both derived piezoelectric composite models are well-posed. Furthermore, we show through analysis and simulations that both current-controlled piezoelectric composites are asymptotically stabilizable through simple electric feedback, which renders the system passive in a classical way for certain system parameters. In this work, we also review several related piezoelectric beams, actuators, and composite models.      
\end{abstract}


\begin{acronym}
    \acro{1D}{one\hyp{}dimensional}
    \acro{2D}{two\hyp{}dimensional}
    \acro{A.S.}{Assymptotically Stabilizable}
    \acro{BIBO}{Bounded Input Bounded Output}
    \acro{CbI}{Control by interconnection}
    \acro{EB}{Euler\hyp{}Bernoulli}
    \acro{EBBT}{Euler\hyp{}Bernoulli beam theory}
    \acro{IBP}{integration by parts}
    \acro{IDA}{Interconnection and Damping Assignment}
    \acro{ODE}{Ordinary Differential Equation}
    \acro{ODEs}{Ordinary Differential Equations}
    \acro{PBC}{Passivity\hyp{}based control}
    \acro{PD}{Proportional-Derivative}
    \acro{PDE}{Partial Differential Equation}
    \acro{PDEs}{Partial Differential Equations}
    \acro{pH}{port-Hamiltonian}
    \acro{PI}{Proportional-Integral}
    \acro{PID}{Pro-portional\hyp{}Integral\hyp{}Derivative}
    \acro{PZT}{lead zirconate titanate}
    \acro{T}{Timoshenko}
    \acro{TBT}{Timoshenko beam theory}
\end{acronym}

\section{Introduction}
    A piezoelectric actuator is a piece of piezoelectric material sandwiched between two layers of electrodes. The actuator can be compressed or elongated in one or more directions by applying an electric stimulus, such as voltage, charge, or current \cite{smith2005smart}. A specific type of electric actuator is the piezoelectric beam, where an electric stimulus acting on the transverse axis incurs deformation in the longitudinal direction. The simplest models for piezoelectric beams are similar to wave equations and describe the longitudinal vibrations of the stresses and strains in the piezoelectric material \cite{MenOSIAM2014}. Several variants exist corresponding to different electrical inputs, such as voltage, charge and current. In either case, the model is similar to the well-known wave equation and exhibits some interesting properties making it useful for many applications. By bonding a piezoelectric actuator onto the surface of a mechanical substrate, the deformation of the actuator incurs shear stress in the substrate, resulting in the curvature of the composition. A mechanical substrate with one or more piezoelectric actuators we refer to as a piezoelectric composite and is useful in high-precision applications.\\

	From a control perspective, two types of applications exist for piezoelectric composites: vibration control and shape control. Vibration control finds applications in acoustic devices \cite{Ralib2015} or suppression of vibrations in mechanical systems \cite{Samikkannu2002vibrationdamp}. Shape control includes applications such as flexible wings \cite{CHUNG2009136}, inflatable space structures \cite{voss2010port}, and deformable mirrors \cite{Radzewicz04}. Often, in applications including inflatable space structures and deformable mirrors, one side of the substrate has a specific function (e.g. reflecting electromagnetic waves). Therefore, in this work, we focus on composites where the piezoelectric actuator is attached to one side of the purely mechanical layer; see Fig \ref{ch3:fig:piezoelectriccomposite} for a depiction.\\
	
	\begin{figure}
		\centering
		\includegraphics[width=\columnwidth]{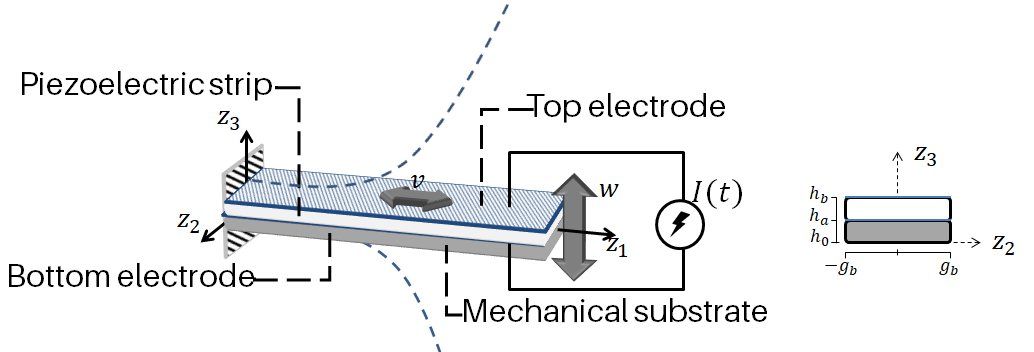}
		\caption[Current-actuated piezoelectric composite]{The piezoelectric composite consists of a piezoelectric actuator (top) bonded to a mechanical substrate (bottom). The piezoelectric actuator consists of a piezoelectric strip sandwiched by a top and bottom electrode. Whereas the piezoelectric actuator can deflect in the longitudinal direction $(z_1)$, the piezoelectric composite can deflect both longitudinally $v$ and transversely $w$ due to the shear stresses between the two layers.}
		\label{ch3:fig:piezoelectriccomposite}
	\end{figure}

	The dynamics of a piezoelectric composite are governed by a set of partial differential equations (PDEs) that originate from continuum mechanics in the mechanical domain and Maxwell's equations in the electromagnetic domain. PDEs with distinct structures and properties are derived by changing the mechanical or electromagnetic domain assumptions. Various beam theories can be assumed for the mechanical domain, and complex non-linear phenomena can be incorporated. Different treatments of Maxwell's equations for the electromagnetic domain result in different dynamics and coupling of the electric, magnetic, and mechanical quantities. Often, in literature, the assumptions for the electromagnetic domain are typified as a \textit{static electric field}, a \textit{quasi-static electric field}, or a \textit{fully dynamic electromagnetic field}. Recent efforts show that the controllability and stabilizability of the models of piezoelectric beams and composites can be significantly altered by choice of input and the treatment of the electromagnetic domain \cite{MenOSIAM2014,VenSSIAM2014,Ozkan2019comp}.\\
	
    The traditional choice of actuation is voltage. Voltage-actuated linear infinite-dimensional piezoelectric beam models are exactly controllable and exponentially stabilizable in the case of a static or quasi-static electric field, see \cite{komornik2005}, \cite{KapitonovMM07}, and \cite{LASIECKA2009167}. In \cite{MenOSIAM2014}, the fully dynamic voltage actuated beam model is asymptotically stabilizable for almost all system coefficients and exponentially stabilizable for a small set of system coefficients. In \cite{OzerMCSS2015}, it has been shown that the fully dynamical beam is polynomially stabilizable when certain conditions on the physical coefficients are satisfied.\\
	
	Due to the less-hysteretic behaviour of charge and current actuated piezoelectric systems, recent studies involved the stabilizability of such models. The charge-actuated models show similar stabilizability properties to voltage-actuated systems due to their duality and the fact that both have boundary inputs. There are two ways to derive current actuated piezoelectric models. {The simplest is obtained by adding a dynamical equation on the boundary to convert the charge input into a current input, which we will refer to as \textit{current-through-the-boundary} actuation \cite{OverviewpaperdeJong_arxiv}. Physically, this corresponds to incorporating some electric circuitry.} The other way of obtaining current actuated systems is by careful considerations for the electromagnetic domain, resulting in a \textit{purely} current actuated system. The modelling of purely current actuated piezoelectric systems is described in \cite{MenOSIAM2014,OzerTAC2019}.\\

	The stabilizability of current-actuated piezoelectric beams has been investigated under a fully dynamic electromagnetic field. The model is derived using magnetic vector potentials that require an additional gauge condition to guarantee a solution in \cite{MenOCDC2014}. It has been shown that the control input is bounded in the energy space and can utmost asymptotically stabilize the system. Recently, in \cite{OzerTAC2019}, the current actuated piezoelectric beam model from \cite{MenOCDC2014} has been interconnected with a substrate. Due to magnetic vector potentials, the model allows for both current and charge input. It has been shown that the reduced electrostatic models with charge and current-through-the-boundary models are respectively exponentially and asymptotically stabilizable, and the derived purely current actuated fully dynamical model is not stabilizable, which is against intuition.\\
	
	In \cite{VenSSIAM2014}, two current actuated non-linear piezoelectric composite systems are derived from a port-Hamiltonian \cite{port-HamiltonianIntroductory14} perspective. One system is derived as a purely current-actuated system with a quasi-static electric field assumption through interconnection, and the other is a current-through-the-boundary system with a fully dynamic electromagnetic field assumption. The approximations of the fully dynamic piezoelectric beam \cite{VenSSIAM2014}, using the mixed finite element method \cite{GoloSchaft2004}, has been shown in \cite{voss2011CDC} to satisfy a necessary condition for asymptotic stabilizability, whereas the quasi-static systems lack this condition. \\

    \blue{The purely current actuated piezoelectric fully dynamical models from \cite{MenOCDC2014} and \cite{OzerTAC2019} are derived with the use of the electric and magnetic potentials, related to the electric field by use of Gauss's magnetic law \eqref{ch3:eq:GausssMagnetic}. The purely current actuated quasi-static piezoelectric model from \cite{VenSSIAM2014} and \cite{voss2011CDC} is derived by imposing an algebraic constraint on the electric field, reducing the electromagnetic coupling to the quasi-static electric field assumption using different quantities to express the reduced electromagnetic domain. The coupled dynamics are obtained through the passive interconnection of the quasi-static dynamical equation and the mechanical beam equations in a similar fashion as the fully dynamical current-through-the-boundary model in \cite{VenSSIAM2014}, which uses electric flux and charge to describe the electromagnetic dynamics. The extensive version, which includes the electromagnetic coupling of the quasi-static piezoelectric beam and composite model described in \cite{VenSSIAM2014} is, to the best of our knowledge, not present in the literature and is helpful to determine the relations to the quasi-static and fully dynamical current actuated piezoelectric models and their stabilization properties.
 }

\subsection{Contributions}
 
    This work presents a novel, fully dynamic piezoelectric composite model with purely current input. We use, to the best of our knowledge, a new treatment of the electromagnetic domain that does not require the use of a gauge function. We show that the new model is the more extensive version which includes the fully electromagnetic coupling of the quasi-static electromagnetic model in \cite{voss2011CDC} and \cite{VenSSIAM2014}.  We consider both the Euler-Bernoulli and Timoshenko beam theories for the mechanical domain and investigate the well-posedness of the derived piezoelectric composite models. Finally, we investigate the stabilizability of the derived piezoelectric composite models and accompany our results with simulations.\\
	
	
	
	The set-up of this paper allows us to present a comprehensive modelling framework for current controlled piezoelectric composites and presents the following key contributions:
    \begin{enumerate} 
        \item A novel modelling approach is developed to capture the dynamics of the electromagnetic domain for current controlled piezoelectric composites and actuators by using a \textit{combined Lagrangian} composed of a mechanical Lagrangian and electromagnetic co-Lagrangian which are coupled through non-energetic elements, known as traditors \blue{\cite{Duinkers1959Traditors}}. 
        \item Two well-posed current-controlled piezoelectric composite models are derived with fully dynamic electromagnetic field, using the novel modelling approach for different beam theories to capture the behaviour of the mechanical domain, i.e. use the Euler-Bernoulli and Timoshenko beam theory in the different models.
        \item It is proved that the derived fully dynamic electromagnetic current-controlled piezoelectric composite models are asymptotically stabilizable for certain system parameters.
        \item The asymptotically stabilizing behaviour of the closed-loop system obtained through classical passivity techniques is illustrated through simulations. 
    \end{enumerate}
    
\subsection{Outline}
    In Section \ref{ch3:sec:modelderivation}, we treat the derivation of the novel piezoelectric composite models for both the Euler-Bernoulli and Timoshenko beam theory and compare the treatment of Maxwell's equations to the existing modelling approaches for piezoelectric material. In Section \ref{ch3:SEC:well-posedness}, we show that both derived piezoelectric composites' are well-posed with the use of semigroup theory \cite{CurtainZwart1995introduction} and in Section \ref{ch3:sec:comparison_beamandcomp}, we compare the derived models to existing piezoelectric beam and composite models. Furthermore, in Section \ref{ch3:sec:Stabz}, we investigate the asymptotic stabilizability properties of the approximated composites and provide some illustrative simulations to accompany the stabilizability results in Section \ref{ch3:sec:simulations}. Finally, in Section \ref{ch3:sec:discussions}, we give some concluding remarks and future research directions.

\section{Model derivation of current controlled  piezoelectric actuators and composites}\label{ch3:sec:modelderivation}
The piezoelectric composite model depicted in Fig \ref{ch3:fig:piezoelectriccomposite} is composed of two layers, which we consider to be perfectly bonded. The top layer of the composite is the piezoelectric actuator, and the bottom layer is a purely mechanical substrate which we denote respectively using the subscript $p$ and $s$. For both layers we consider a volume $\Omega$ with length $\ell$, width $2g_b$ , and thickness $h=h_b-h_a$ in the Cartesian coordinate system $z_1,z_2,z_3$ with unit vectors $(\vect{z_1},\vect{z_2},\vect{z_3})$, as depicted in  Fig \ref{ch3:fig:piezoelectriccomposite}. Let $k\in\{p,s\}$ then the body $\Omega^k$ of these layers can be defined as follows
\begin{align*}
&\Omega^k:= \left\{ (z_1,z_2,z_3)\ \mid  0\leq z_1\leq \ell, \right. \\ & \quad\left.-g_b\leq z_2 \leq g_b,\ h^k_a\leq z_3 \leq h^k_b\right\} ,
\end{align*}
where we assume that the length $\ell$ is significantly larger than the width and thickness of the volume and we have that $h^s_a=h_0=-h_1$, $h^s_b=h_1=h^p_a$, and $h^p_b=h_2$. We denote the longitudinal deformation along $z_1$ by $v(z_1,t)$, the transverse deformation along $z_3$ by $w(z_1,t)$, and the rotation of the beam given by $\phi(z_1,t)$. Furthermore, denote the strain $\vect{\epsilon}=\col(\epsilon_{11}, \epsilon_{22} , \epsilon_{33}, \epsilon_{23}, \epsilon_{31}, \epsilon_{21})$, stress $\vect{\sigma}=\col(\sigma_{11}, \sigma_{22}, \sigma_{33}, \sigma_{23}, \sigma_{31}, \sigma_{21})$, electric displacement $\vect{D}=\col(D_1,D_2,D_3)$, and the electric field $\vect{E}=\col(E_1,E_2,E_3)$ and consider the linear piezoelectric constitutive relations \cite{IEEEstandardPiezo,tiersten1969linear}, coupling the mechanical and electromagnetic domain as follows,
 \begin{alignat}{1}\label{ch3:eq:Constitutive_relations}
    	\left[\begin{array}{c}
            	\vect{\sigma}\\
            	\vect{D}
        	\end{array}\right]= & \left[\begin{array}{cc}
            	C^E & -e\\
            	e^T & \varepsilon
            	\end{array}\right]\left[\begin{array}{c}
            	\vect{\epsilon}\\
            	\vect{E}
        	\end{array}\right],
    	\end{alignat}
where $C^E$ is the $6\times6$ stiffness matrix, $e$ denotes the $6\times 3$ piezoelectric constants matrix, and $\varepsilon$ is the $3\times 3$ diagonal permittivity matrix \cite{tiersten1969linear,IEEEstandardPiezo,preumont_piezoelctric2006a}. The constitutive relations \eqref{ch3:eq:Constitutive_relations} take care of the coupling between the mechanical and electromagnetic domains. For the purely mechanical substrate, the electromechanical coupling is not present, i.e. $e=[0]$. Let's denote the stiffness coefficient, shear modulus, piezoelectric coefficient, and permittivity constant by $C^{k,E}_{11}, G:=C^{k,E}_{55}, \gamma:=e_{15}, \varepsilon_{33}$, respectively. For linear isotropic piezoelectric dielectric material with polarization in the $z_3-$direction, i.e. $E_1,E_2=0$ we obtain the  displacement field $\vect{u}$, strain and constitutive relations for the Euler Bernoulli beam theory (EBBT) \cite{carrera2011beam} as follows,
\begin{subequations}\label{ch3:eq:EBBT_equations}
\begin{align}
    \begin{split}\label{ch3:eq:disp_EBBT_equations}
        u^k_1(z_1)&=v(z_1)-\vect{z_3}\frac{\partial}{\partial  z_1}w(z_1)\\
        u^k_3(z_1)&=w(z_1),\\
    \end{split}\\
     \begin{split}\label{ch3:eq:strain_EBBT_equations}
        \epsilon^k_{11}&=\frac{\partial}{\partial z_1}v(z_1)-\vect{z_3}\frac{\partial^2}{\partial z_1^2}w(z_1),\\
    \end{split}\\
     \begin{split}\label{ch3:eq:const_EBBT_equations}
        \sigma^s_{11}&=C_{11}^{s,E}\epsilon_{11}\\
         \sigma^p_{11}&=C_{11}^{p,E}\epsilon_{11}-\gamma E_3\\
        D_3&=\varepsilon_{33} E_3 + \gamma \epsilon_{11}.
    \end{split}
\end{align}
\end{subequations}
For the Timoshenko beam theory (TBT) \cite{carrera2011beam} the  displacement field, strain and constitutive relations are as follows,
\begin{subequations}\label{ch3:eq:TBT_equations}
\begin{align}
    \begin{split}\label{ch3:eq:disp_TBT_equations}
         u_1^k(z_1)&=v(z_1)-\vect{z}_3\phi(z_1)\\
        u_3^k(z_1)&=w(z_1),\\
    \end{split}\\
     \begin{split}\label{ch3:eq:strain_TBT_equations}
        \epsilon^k_{11}&=\frac{\partial}{\partial z_1}v(z_1)-\vect{z_3}\frac{\partial}{\partial z_1}\phi(z_1)\\
        \epsilon_{13}^k&=\frac{1}{2}\braces{\frac{\partial}{\partial z_1}w(z_1)-\phi(z_1) },
    \end{split}\\
     \begin{split}\label{ch3:eq:const_TBT_equations}
     \sigma^s_{11}&=C_{11}^{s,E}\epsilon_{11}\\
        \sigma^p_{11}&=C_{11}^{p,E}\epsilon_{11}-\gamma E_3\\
        \sigma^k_{13}&=G^k\epsilon_{13}\\
         D_3&=\varepsilon_{33} E_3 + \gamma \epsilon_{11}.
    \end{split}
\end{align}
\end{subequations}
To model the fully dynamic electromagnetic current-controlled piezoelectric composite, we consider two types of energy; the mechanical energies, composed of the kinetic co-energy $(T^{k,\ast})$ and the potential energy $(V^k)$, used for both the mechanical substrate and the piezoelectric actuator; and the electromagnetic energies, composed of the electric energy $(\mathcal{E})$ and magnetic energy $(\mathcal{M})$, used for the piezoelectric actuator. Let $\rho^k$ denote the mass density of the material and let the vectors  $\vect{B}$ and $\vect{H}$ denote the magnetic field and the magnetic field intensity. Then, the energies of the piezoelectric composites are given as follows,
\begin{subequations}\label{ch3:eq:energies_all}
\begin{align}
      T^{k,\ast}&=\frac{1}{2}\int_\Omega \rho^k(\dot{\vect{u}^k} \cdot \dot{\vect{u}^k})~d\Omega, \label{ch3:eq:energy_mechkinetic}\\
    V^k&=\frac{1}{2}\int_\Omega \vect{\sigma^k}\cdot\vect{\epsilon^k}~d\Omega.\label{ch3:eq:energy_mechpotential}\\
        \mathcal{E}&=\frac{1}{2}\int_\Omega  \vect{D}\cdot \vect{E}\ d\Omega, \label{ch3:eq:Energy_Electric} \\
        \mathcal{M}&=\frac{1}{2}\int_\Omega \vect{H}\cdot\vect{B}\ d\Omega,\label{ch3:eq:Energy_Magnetic} 
\end{align}
\end{subequations}
where we omit the spatial dependency (on $z_1$). To describe the behaviour of the fully dynamic electromagnetic field of the piezoelectric actuator we require Maxwell's equations for dielectrics and the piezoelectric constitutive relations for permeable material \cite{eom2013maxwell,tiersten1969linear}.
Therefore, let $\mu$ represent the magnetic permeability of the material, denote the volume charge density by $\sigma_v$, and let $\vect{J}$ denote the free current charges.  
Then, Maxwell's equations  \cite{eom2013maxwell} can be written by the four laws;
        	\begin{subequations}\label{ch3:eq:Maxwells_full}
    	\begin{align}
        		\nabla\times \vect{E}&=-\frac{\partial \vect{B}}{\partial t}, &\,\text{Faraday's law} \label{ch3:eq:Faraday'sLaw} \\
        	\nabla\cdot\vect D&=\sigma_v, &\,\text{Gauss's Electric law}\label{ch3:eq:GausssElectric} \\
        	\nabla\times \vect{H}&=\frac{\partial \vect D}{\partial t}+\vect{J}, &\,\text{Max-Ampere's law}\label{ch3:eq:Ampereslaw} \\
        	\nabla\cdot\vect{B}&=0,  &\,\text{Gauss's Magnetic law}\label{ch3:eq:GausssMagnetic}
    	\end{align}
    	\end{subequations}
     Additionally, we consider the two constitutive relations
    	\begin{subequations}\label{ch3:eq:constrel_Maxwell}
    	\begin{align}
        	\vect{D}&=\varepsilon\vect{E},\label{ch3:eq:constrel_MaxwellD}\\
        	\mu\vect{H}&=\vect{B},\label{ch3:eq:constrel_MaxwellH}
    	\end{align}
    	\end{subequations}
    	for isotropic magnetic permeable material. The formerly presented material is similar to sections of \cite{MenOACC2014} and \cite{Ozkan2019comp}, which look into voltage and current-actuated piezoelectric beams and composites, respectively. Similarly to voltage controlled piezoelectric actuators, we have that $E_3\neq 0$ and $E_1=E_2=0$, which reduces \eqref{ch3:eq:Maxwells_full} and \eqref{ch3:eq:constrel_Maxwell} to scalar equations with $E_3, D_3, H_2$,  $B_2$ $J_3$ as the remaining nonzero physical quantities. \\
    	
    	In this work, we consider two piezoelectric composites where the piezoelectric actuators are actuated by applying an electric current flowing across the $z_3$ direction of the piezoelectric layer. The modelling approach presented here circumvents the need for a gauge function to ensure a well-posed actuator or composite, as opposed to \cite{Ozkan2019comp}. This is accomplished by defining the magnetic flux $\Phi$ as follows,
    	\begin{align}\label{ch3:eq:magneticflux_flow3}
            \Phi(z_1):=\int_{0}^{z_1}{B_2}(\xi)d\xi,
    	\end{align}
    and obtain from  Faraday's law \eqref{ch3:eq:Faraday'sLaw} the relations
    	\begin{align}\label{ch3:eq:magneticfluxexpressions} 
        	\begin{split}
            	\frac{\partial}{\partial z_1}{\Phi}&={B}_2,\\
            	\dot{\Phi}&=E_3
        	\end{split}
    	\end{align}	
    The expressions in \eqref{ch3:eq:magneticfluxexpressions} are useful for writing the
    energies \eqref{ch3:eq:energies_all} in scalar form for to the novel current-controlled piezoelectric actuator and composite model. Furthermore, we consider a combined Lagrangian to derive the dynamical equations using Hamilton's principle \cite{lanczos1970variational}. For electromagnetic systems, it may be possible that due to dissipation or the source, a Lagrangian or co-Lagrangian formulation is not sufficient to derive the equations of motion and the so-called combined Lagrangian formulation is required \cite{JeltsemaMDM2009}, which is the case here due to the current source. Due to the current source, we are dealing with a force balance that needs to be coupled with the mechanical flow balance. The combined Lagrangian is composed of a Lagrangian function coupled with a co-Lagrangian function and takes care of the coupling through non-energetic coupling terms known as traditors \cite{Duinkers1959Traditors}.\\
    
        In classical mechanics, the Lagrangian is defined as the difference between the kinetic (Co-)energy $\mathcal{T}^\ast$ and the total potential energy $\mathcal{V}$. The kinetic co-energy is the dual (i.e. complementary form) of the kinetic energy and is related to each other by the Legendre transformation. \blue{A property of the (kinetic) co-energies (annotated by $^\ast$) is the association with a flow (i.e. the rate of change of the generalised displacements).} Whereas the (potential) energy is associated with the generalised displacement \cite{JeltsemaMDM2009}. Besides the co-energy, we also have the notion of the co-Lagrangian functional $\mathcal{L}^\ast$, which is the dual of the Lagrangian functional $\mathcal{L}$, see Table \ref{ch3:tab:domainLagrangian} for an overview of the composition of the Lagrangian and co-Lagrangian functional for the mechanical and electromagnetic domain \cite{JeltsemaMDM2009}. For linear mechanical systems, we have that the kinetic co-energy and kinetic energy are similar, i.e. $\mathcal{T}^\ast(f)=\mathcal{T}(p)$, where $f$ and $p$ denote the (generalized) velocity and (generalized) momentum, respectively. Therefore, we often do not care about the difference between kinetic energy and kinetic co-energy. However, for multi-domain modelling, the use of energy and co-energy is relevant, especially when a coupling is required between a flow balance and a force balance.\\
        
        The coupling terms, known as traditors \cite{Duinkers1959Traditors}, facilitate the coupling between the velocity and force balance resulting from the respective Lagrangian and co-Lagrangian formulation \cite{JeltsemaMDM2009}. Traditors are characterized by the property that at any moment in time, the total power delivered to these multi-port elements is zero. Therefore, traditors are not present in the total energy function. Two linear examples of traditors are the gyrator and transformer \cite{Duinkers1959Traditors}.   \\    
   
        \begin{table}
        \centering
      \begin{tabular}{|r||r|r|}
      \hline
      Domain & Lagrangian & co-Lagrangian \\
      \hline\hline
      Mechanical  & Kinetic co-energy   & Potential co-energy \\ 
                  & Potential energy    & Kinetic Energy \\ \hline
      Electromagnetic & Magnetic co-energy & Electric co-energy \\
        & Electric energy   & Magnetic energy\\  \hline
    \end{tabular}
    \caption[Lagrangian and co-Lagrangian for mechanical and electromagnetic domains]{Lagrangian and co-Lagrangian for mechanical and electromagnetic domains.}
    \label{ch3:tab:domainLagrangian}
    \end{table}
    
    In our work, we consider two piezoelectric composite models by using two mechanical beam theories, \ac{EBBT} and \ac{TBT}, denoted 
    by using \ac{EB} and \ac{T} as subscripts. The used combined Lagrangian \eqref{ch3:eq:Lagrangian} takes the form,
    $$\mathcal{L}_j= \mathcal{L}_{\text{mech},j}+\mathcal{L}_{\text{em},j}^\ast,$$ 
    where $j\in \setof{\text{EB,T}}$ to differentiate between the different beam theories. The combined Lagrangian is composed of the considered energies \eqref{ch3:eq:energies_all}, which are given in more detail below \blue{ and contains the non-energetic coupling terms (traditors)}.
    
    Therefore, define the cross-section and inertia corresponding to the piezoelectric actuator and mechanical substrate layers as
    	\begin{align}\label{ch3:eq:crossseactionalA}
        	A^k&:=\int_{h_a}^{h_b}\int_{-g_b}^{g_b}dz_2dz_3=2g_b(h_b-h_a),\nonumber\\
        	I^k&:=\int_{h_a}^{h_b}\int_{-g_b}^{g_b}{z}_3^2dz_2dz_3=\frac{2}{3}g_b(h_b^3-h_a^3),\\ 
        	I_0^k&:=\int_{h_a}^{h_b}\int_{-g_b}^{g_b}\vect{z}_3dz_2dz_3=\left[g_b{z}_3^2\right]_{h_a}^{h_b}=g_b(h_b^2-h_a^2).\nonumber \\ 
        	A_q&:=\int_0^\ell dz_2 dz_1=\int_0^\ell 2g_b dz_1,\nonumber
	\end{align}
    where $A_q$ denotes the surface in the $z_1 z_2$ plane where the charges flow through. Note that $I_0^s=0$ since $h^s_b=-h^s_a$.\\
    
 \textbf{\ac{EBBT} energies:}\\
 Now we can give the explicit forms of the energies \eqref{ch3:eq:energies_all} of the layers of the  piezoelectric composite, using the \ac{EBBT} as follows,
\begin{subequations}\label{ch3:eq:EBBTenergies_all_voluit}
\begin{align}
      T&^{k,\ast}_{EB}=\frac{1}{2}\rho^k \int_\Omega \dot{u}_1^2+\dot{u}_3^2~ d\Omega  \label{ch3:eq:EBBTenergy_mechkinetic_scalar}  \\ 
       &=\frac{1}{2}\int_{0}^{\ell}\brackets{\rho^k\braces{A^k\dot{v}^2-2I^k_0\dot{v}\dot{w}_{z_1}+I^k\dot{w}^2_{z_1}+A^k\dot{w}^2}}dz_1, \nonumber\\
      V&^s_{EB}=\frac{1}{2}\int_\Omega \sigma^s_{11}\epsilon^s_{11}~d\Omega \label{ch3:eq:EBBTenergy_mechpotential_scalar_s}  \\
       &=\frac{1}{2}\int_{0}^{\ell}\left[C^s_{11}\braces{A^s v_{z_1}^2+I^s w_{z_1z_1}^2}\right. \nonumber \\
    V&^p_{EB}=\frac{1}{2}\int_\Omega \sigma_{11}\epsilon_{11}~d\Omega \label{ch3:eq:EBBTenergy_mechpotential_scalar_p}  \\
       &=\frac{1}{2}\int_{0}^{\ell}\left[C^p_{11}\braces{A^p v_{z_1}^2+I^p w_{z_1z_1}^2-2I^p_0v_{z_1}w_{z_1z_1}}\right. \nonumber \\
       &\quad \left.-\gamma \braces{A^p v_{z_1}-I_0^p w_{z_1z_1}}\dot{\Phi}\right]dz_1,\nonumber \\
    \mathcal{E}&^\ast_{EB}=\frac{1}{2}\int_\Omega  {D}_3 {E}_3\ d\Omega, \label{ch3:eq:EBBTEnergy_Electric_scalar}  \\
     &=\frac{1}{2}\int_{0}^{\ell}\brackets{\varepsilon_{33}A^p\dot{\Phi}^2+\gamma \braces{A^p v_{z_1}-I^p_0w_{z_1z_1}}\dot{\Phi} }dz_1, \nonumber\\
      \mathcal{M}&=\frac{1}{2}\int_{\Omega}{\frac{1}{\mu}B_2^2}d\Omega=\frac{1}{2}\int_{0}^{\ell}\brackets{\frac{1}{\mu}A^p\Phi_{z_1}^2}dz_1,\label{ch3:eq:EBBTEnergy_Magnetic_scalar}
		\end{align}
\end{subequations}

where we made use of \eqref{ch3:eq:magneticflux_flow3}, \eqref{ch3:eq:EBBT_equations} and \eqref{ch3:eq:crossseactionalA}.

 \textbf{\ac{TBT} energies:}\\
 The energies for the piezoelectric actuator using T beam theory are as follows,
\begin{subequations}\label{ch3:eq:TBTenergies_all_voluit}
\begin{align}
      T^{k,\ast}_{T}&=\frac{1}{2}\rho^k \int_\Omega \dot{u}_1^2+\dot{u}_3^2~ d\Omega  \label{ch3:eq:TBTenergy_mechkinetic_scalar}  \\ 
       &=\frac{1}{2}\int_{0}^{\ell}\brackets{\rho^k\braces{A^k\dot{v}^2-2I^k_0\dot{v}\dot{\phi}+I^k\dot{\phi}^2+A^k\dot{w}^2}}dz_1, \nonumber\\
     V^s_{T}&=\frac{1}{2}\int_\Omega \sigma_{11}\epsilon_{11} +\sigma_{13}\epsilon_{13}~d\Omega  \label{ch3:eq:TBTenergy_mechpotential_scalar_s}  \\
       &=\frac{1}{2}\int_{0}^{\ell}\left[C^{s,E}_{11}\braces{A^sv_{z_1}^2+I^s\phi_{z_1}^2}\right. \nonumber \\
       &\quad\left. +\frac{1}{2}A^sG^s\braces{w_{z_1}^2-2w_{z_1}\phi+\phi^2}\right]dz_1,\nonumber \\ 
     V^p_{T}&=\frac{1}{2}\int_\Omega \sigma_{11}\epsilon_{11} +\sigma_{13}\epsilon_{13}~d\Omega  \label{ch3:eq:TBTenergy_mechpotential_scalar_p}  \\
       &=\frac{1}{2}\int_{0}^{\ell}\left[C^{p,E}_{11}\braces{A^pv_{z_1}^2+I^p\phi_{z_1}^2-2I^p_0v_{z_1}\phi_{z_1}}\right. \nonumber \\
       &\quad \left.-\gamma\braces{A^p v_{z_1}-I^p_0\phi_{z_1}}\dot{\Phi}\right.\\
       &\quad\left. +\frac{1}{2}A^p G^p\braces{w_{z_1}^2-2w_{z_1}\phi+\phi^2}\right]dz_1,\nonumber \\
    \mathcal{E}^\ast_{T}&=\frac{1}{2}\int_\Omega  {D}_3 {E}_3\ d\Omega, \label{ch3:eq:TBTEnergy_Electric_scalar}  \\
     &=\frac{1}{2}\int_{0}^{\ell}\brackets{\varepsilon_{33}A^p\dot{\Phi}^2+\gamma\braces{A^p v_{z_1}-I^p_0\phi_{z_1}}\dot{\Phi} }dz_1, \nonumber
		\end{align}
\end{subequations}
	where we made use of \eqref{ch3:eq:magneticflux_flow3}, \eqref{ch3:eq:TBT_equations} and \eqref{ch3:eq:crossseactionalA} and $\mathcal{M}$ for \ac{TBT} is the same as \eqref{ch3:eq:EBBTEnergy_Magnetic_scalar}.\\

Following Table \ref{ch3:tab:domainLagrangian} and considering \eqref{ch3:eq:magneticflux_flow3} as the generalized displacements coordinate for the electromagnetic part, we see that \eqref{ch3:eq:EBBTEnergy_Electric_scalar} and  \eqref{ch3:eq:TBTEnergy_Electric_scalar} can be regarded as the electric co-energies, for \ac{EBBT} and \ac{TBT} respectively. Furthermore, we have the magnetic energy \eqref{ch3:eq:EBBTEnergy_Magnetic_scalar} for the electromagnetic domain, which is not influenced by the beam theory. Therefore, we have a combined Lagrangian for the novel current-controlled piezoelectric actuator, composed of the mechanical Lagrangian and electromagnetic co-Lagrangian, as follows
\begin{align}\label{ch3:eq:Lagrangian}
    \mathcal{L}_j &=  (T^{s,\ast_j} + T^{p,\ast_j)} - (V^s_j +V^p_j) +  \mathcal{E}^\ast_j -  \mathcal{M}\nonumber \\
    &= (T^\ast_j + \mathcal{E}^\ast_j) - (V_j + \mathcal{M}),
\end{align}
\blue{where the non-energetic coupling terms are present in $V^p_j$ and $\mathcal{E}^\ast_j$.} To derive the dynamical models for the piezoelectric actuators using \ac{EBBT} and \ac{TBT}, we make use of \eqref{ch3:eq:Lagrangian}.\\

The current source $\mathcal{I}(t)$ enters the system through the external work term $W$
	\begin{align}\label{ch3:eq:workterm_current}
		\begin{split}
		W &=-\int_{\Omega}\sigma_a E_3 d\Omega= \int_{\Omega}\sigma_a \dot{\Phi} d\Omega\\
		&=-\int_{\Omega}\dot{\sigma}_a\Phi d\Omega=-\int_{\Omega}\frac{1}{A_q}\mathcal{I}(t)\Phi d\Omega\\
		&=-\int_0^\ell (h_b-h_a)\mathcal{I}(t)\Phi dz_1, 
	\end{split}
	\end{align}
where we made use of \eqref{ch3:eq:magneticfluxexpressions} and integration by parts on the standard work term of voltage actuated piezoelectric actuators \cite{tiersten1969linear,IEEEstandardPiezo,smith2005smart}. \\

The mathematical models of the voltage actuated piezoelectric actuator can be derived by applying Hamilton's principle to a definite integral that contains the Lagrangian $\mathcal{L}$ and the impressed forces $W$ that allows the actuation of the piezoelectric actuator by means of an applied current. Lets define the definite integral 
\begin{align}\label{ch3:eq:def_integrap_heli}
    \mathcal{J}_j& :=\int_{t_{0}}^{t_{1}} \mathcal{L}_j + W ~d t.
\end{align}
Furthermore, define the coefficients
    	\begin{eqnarray}\label{ch3:eq:new_energycoefficients}
        		&\rho_A:=\rho^p A^p+\rho^s A^s, \quad &C_A:=C^{p,E} A^p+C^{s,E} A^s, \nonumber \\
        		&\rho_I:=\rho^p I^p+\rho^s I^s, \quad  &C_I:=C^{p,E} I^p+C^{s,E} I^s, \nonumber  \\
        		&\rho_{I_0}:=\rho^p I^p_0,  \quad &C_{I_0}:=C^p_{11}I^p_0, \nonumber\\
        		&G_A :=G^s A^s + G^p A^p 
    	\end{eqnarray} 
such that we can write the definite integrals \eqref{ch3:eq:def_integrap_heli} for \ac{EBBT} and \ac{TBT}, composed of the energies \eqref{ch3:eq:EBBTenergies_all_voluit} and \eqref{ch3:eq:TBTenergies_all_voluit} of the perfectly bonded piezoelectric actuator and mechanical substrate and external-work \eqref{ch3:eq:workterm_current} respectively as follows,
\begin{align} \label{ch3:eq:def_integral_EBBT}
        \mathcal{J}_{EB}&=\int^{t_1}_{t_0}\frac{1}{2}\int_0^\ell \left[\braces{\rho_A\dot{v}^2-2\rho_{I_0}\dot{v}\dot{w}_{z_1}+\rho_I\dot{w}^2_{z_1}+\rho_A\dot{w}^2} \right. \nonumber\\
         &+\left. \varepsilon_{33}A^p\dot{\Phi}^2+2\gamma\braces{A^pv_{z_1}-I^p_0w_{z_1z_1}}\dot{\Phi}\right.\nonumber \\
        &  \left. -\braces{C_Av_{z_1}^2-C_I w_{z_1z_1}^2-2C_{I_0}v_{z_1}w_{z_1z_1}} \right. \nonumber \\
       & \left. - \frac{1}{\mu}A^p\Phi_{z_1}^2 - 2 (h_2-h_1) \mathcal{I}(t)\Phi \right] dz_1 ~dt,
\end{align}
and
\begin{align} \label{ch3:eq:def_integral_TBT}
          \mathcal{J}_{T}&=\int^{t_1}_{t_0}\frac{1}{2}\int_0^\ell \left[\braces{\rho_A\dot{v}^2-2\rho_{I_0}\dot{v}\dot{\phi}+\rho_I\dot{\phi}^2+\rho_A\dot{w}^2} \right. \nonumber\\
         &+\left. \varepsilon_{33}A^p\dot{\Phi}^2+2\gamma\braces{A^pv_{z_1}-I^p_0\phi_{z_1}}\dot{\Phi}\right. \nonumber\\
        &  \left. -\braces{C_Av_{z_1}^2+C_I\phi_{z_1}^2-2C_{I_0}v_{z_1}\phi_{z_1}} \right. \nonumber \\
    &    \left. -\frac{1}{2} G_A \braces{w_{z_1}^2-2w_z{z_1}\phi+\phi^2} \right. \nonumber \\
       & \left. - \frac{1}{\mu}A^p\Phi_{z_1}^2 -   2(h_2-h_1) \mathcal{I}(t)\Phi \right] dz_1 ~dt,
\end{align}
\blue{where the non-energetic coupling terms can be recognized as they contain the multiplication of variables originating from both the mechanical and electromagnetic domains.}

\textbf{\ac{EBBT} model:}\\
Application of Hamilton's principle \cite{lanczos1970variational} to \eqref{ch3:eq:def_integral_EBBT} and setting the variation of admissible displacements $\setof{v,w_{z_1},\Phi}$ to zero, yields the  dynamical model describing the behaviour of a voltage actuated piezoelectric actuator with fully dynamic piezoelectric actuator with \ac{EBBT} as follows,
	\begin{align}\label{ch3:eq:PDE_currentactuator_EBBT_FD}
    	\rho_A\ddot{v}-\rho_{I_0}\ddot{w}_{z_1}&=C_Av_{z_1z_1}-C_I w_{z_1z_1z_1}-\gamma A^p\dot{\Phi}_{z_1} \nonumber\\
    	\rho_I\ddot{w}_{z_1}-\rho_{I_0}\ddot{v}&=C_Iw_{z_1z_1z_1}-C_{I_0}v_{z_1z_1}+\gamma C_{I_0}\dot{\Phi}_{z_1} \nonumber\\
    	\epsilon_{33}A^p\ddot{\Phi}&=\frac{1}{\mu}A^p\Phi_{z_1z_1}-\gamma\braces{A^p\dot{v}_{z_1}-I^p_0\dot{w}_{z_1z_1}} \nonumber\\
    	&-\braces{h_2-h_1}\mathcal{I}(t),
    	\end{align}
    	on the spatial domain $z_1\in[0,\ell]$, with essential boundary conditions $v(0)=\dot{v}(0)=0$, $w_{z_1}(0)=\dot{w}_{z_1}(0)=0$, $\Phi(0)=\dot{\Phi}(0)=0$ and natural boundary conditions $C_A v_{z_1}(\ell)-C_{I_0}w_{z_1z_1}(\ell)-\gamma A^p \dot{\Phi}(\ell)=0$, $C_Iw_{z_1z_1}(\ell)-C_{I_0}v_{z_1}(\ell)+\gamma I^p_0  \dot{\Phi}(\ell)=0$, $\frac{1}{\mu}A^p\Phi_{z_1}(\ell)=0$. The total energy of the actuator is given by
    	\begin{align}\label{ch3:eq:Hamiltonian_current_FD_EBBT}
    	\mathcal{H}_{EB}(t)&=\frac{1}{2}\int_{0}^{\ell}\left[\rho_A\dot{v}^2+\rho_I\dot{w}_{z_1}^2-2\rho_{I_0}\dot{v}\dot{w}_{z_1}+\epsilon_{33}\dot{\Phi}^2\right. \nonumber\\ 
    	&\quad\left.+ C_Av_{z_1}^2+C_Iw_{z_1z_1}^2-2 C_{I_0}v_{z_1}w_{z_1z_1}\right. \nonumber\\
    	&\left.+\frac{1}{\mu}A^p \Phi_{z_1}^2\right]dz_1,
    	\end{align}
    	obtained by the Legendre transformation on the Lagrangian  \eqref{ch3:eq:Lagrangian}. \blue{In \eqref{ch3:eq:Hamiltonian_current_FD_EBBT}, the non-energetic coupling terms are absent as expected.}\\
     
\textbf{\ac{TBT} model:}\\    	
Applying Hamilton's principle \cite{lanczos1970variational} to \eqref{ch3:eq:def_integral_TBT} and setting the variation of admissible displacements $\setof{v,\phi,\Phi}$ to zero, yields the  dynamical model describing the behaviour of a voltage actuated piezoelectric composite with fully dynamic piezoelectric actuator with \ac{TBT} as follows,
	\begin{align}\label{ch3:eq:PDE_currentactuator_TBT_FD}
    	\rho_A\ddot{v}-\rho_{I_0}\ddot{\phi}&=C_A v_{z_1z_1}-C_{I_0}\phi_{z_1z_1}-\gamma A^p \dot{\Phi}_{z_1} \nonumber\\
    	\rho_I\ddot{\phi}-\rho_{I_0}\ddot{v}&=C_I\phi_{z_1z_1}-C_{I_0}v_{z_1z_1}+\gamma I_0^p \dot{\Phi}_{z_1} \nonumber\\
    	&\quad +\tfrac{1}{2}G_A\braces{w_{z_1}-\phi} \nonumber\\
    	\rho_A \ddot{w} & = \tfrac{1}{2}G_A\braces{w_{z_1z_1}-\phi_{z_1}}\nonumber \\
    	\varepsilon_{33}A^p \ddot{\Phi}&=\frac{1}{\mu}A^p\Phi_{z_1z_1}-\gamma\braces{A^p\dot{v}_{z_1}-I_0^p\dot{\phi}_{z_1}} \nonumber\\
    	&-\braces{h_2-h_1}\mathcal{I}(t),
    	\end{align}
    	on the spatial domain $z_1\in[0,\ell]$, with essential boundary conditions $v(0)=\dot{v}(0)=0$, $w(0)=\dot{w}(0)=0$, $\phi(0)=\dot{\phi}(0)=0$, $\Phi(0)=\dot{\Phi}(0)=0$ and natural boundary conditions $C_A v_{z_1}(\ell)-C_{I_0}\phi_{z_1}(\ell)-\gamma A^p \dot{\Phi}(\ell)=0$, $C_I\phi_{z_1}(\ell)-C_{I_0}v_{z_1}(\ell)+\gamma I_0^p  \dot{\Phi}(\ell)=0$, $G_A(w_{z_1}-\phi)=0$,  $\frac{1}{\mu}A\Phi_{z_1}(\ell)=0$. The total energy of the actuator is given by
    	\begin{align}\label{ch3:eq:Hamiltonian_current_FD_TBT}
    	\mathcal{H}_{T}(t)&=\frac{1}{2}\int_{0}^{\ell}\left[\rho_A\dot{v}^2+\rho_I\dot{\phi}^2-2\rho_{I_0}\dot{v}\dot{\phi}+ \rho_A \dot{w}^2 \right. \nonumber\\ 
    	&\quad\left. +\frac{1}{\beta}A^p\dot{\Phi}^2+ C_A v_{z_1}^2+C_I\phi_{z_1}^2-2C_{I_0}v_{z_1}\phi_{z_1}\right. \nonumber\\
    	&\left.\quad +\frac{1}{2}AG\braces{w_{z_1}^2-2w_{z_1}\phi+\phi^2}+\frac{1}{\mu}A^p \Phi_{z_1}^2\right]dz_1,
    	\end{align}
    	obtained by the Legendre transformation on the Lagrangian or by summating the energies \eqref{ch3:eq:TBTenergies_all_voluit}. \blue{Furthermore, as anticipated, in \eqref{ch3:eq:Hamiltonian_current_FD_TBT}, the non-energetic coupling terms are not present.}\\

      \blue{ The use of the combined Lagrangian makes the derivation of the fully dynamic electromagnetic current-controlled models \eqref{ch3:eq:PDE_currentactuator_EBBT_FD} and \eqref{ch3:eq:PDE_currentactuator_TBT_FD} possible. Furthermore, a physical interpretation can be given to the coupling components present in the Lagrangian which are not present in the total energies \eqref{ch3:eq:Hamiltonian_current_FD_EBBT} and \eqref{ch3:eq:Hamiltonian_current_FD_TBT}, through the existence of so-called \textbf{traditors} \cite{Duinkers1959Traditors}. In the case of the novel current actuated piezoelectric composites, we are dealing with a gyrator type traditor with gyrating constant $\gamma$ present in the combined Lagrangian $\mathcal{L}_k$ and also in \eqref{ch3:eq:def_integral_EBBT} and \eqref{ch3:eq:def_integral_TBT}, used for deriving the mathematical models.}\\

     \begin{remark}     
            The \ac{EBBT} model \eqref{ch3:eq:PDE_currentactuator_EBBT_FD} and corresponding total energy \eqref{ch3:eq:Hamiltonian_current_FD_EBBT} are easily obtained from the \ac{TBT} model by restricting the rotation to remain perpendicular to the neutral of the beam, i.e. by substituting the constraint $\phi= w_{z_1}$ in \eqref{ch3:eq:PDE_currentactuator_TBT_FD} and total energy \eqref{ch3:eq:Hamiltonian_current_FD_TBT}, respectively.
        \end{remark}

         The distributed current source $\mathcal{I}(t)$ in the proposed models act on the surface of the piezoelectric layer where the electrodes (with surface $A_q$) are located. The surface $A_q$ is in the $z_1z_2$ plane, see Fig \ref{ch3:fig:piezoelectriccomposite}, and the applied current acts on the normal of $A_q$, i.e. in the $z_3$ direction. Therefore, the current density allows a current input through
    \begin{align}\label{ch3:eq:current_actuation}
        J_3=\lim\limits_{A_q\to 0}A_q\mathcal{I}(t).
    \end{align}
    The external work \eqref{ch3:eq:workterm_current} is a direct consequence from Ampere's law \eqref{ch3:eq:Ampereslaw} which becomes evident by reducing \eqref{ch3:eq:Ampereslaw} for the piezoelectric actuator to obtain the scalar equation,    
	    \begin{align}\label{ch3:eq:maxamp_novel_scalar}
		\varepsilon_{33}\frac{\partial^2}{\partial t^2}\Phi&=\frac{1}{\mu}\frac{\partial^2}{\partial z_1^2}\Phi-\gamma\frac{\partial}{\partial t}(v_{z_1}-z_3w_{z_1z_1})-J_3,
	\end{align}
    where we made use of the magnetic flux density  \eqref{ch3:eq:magneticfluxexpressions} and \eqref{ch3:eq:magneticflux_flow3}. Combining \eqref{ch3:eq:maxamp_novel_scalar}, \eqref{ch3:eq:current_actuation} and  integrating both sides with respect to the cross-section $A$, see \eqref{ch3:eq:crossseactionalA}, we obtain the same expression as the third equation of \eqref{ch3:eq:PDE_currentactuator_EBBT_FD}, ensuring the validity of the electromagnetic part in \eqref{ch3:eq:PDE_currentactuator_EBBT_FD}, which is derived using \eqref{ch3:eq:workterm_current}. This can be done similarly for the \ac{TBT} model with the appropriate strain expression.\\

    In the next part we show the well-posedness of the novel fully dynamic electromagnetic current-controlled piezoelectric composites, derived using the Euler-Bernoulli and Timoshenko Beam Theory.
   
\section{Well-posedness of piezoelectric Composites}\label{ch3:SEC:well-posedness}
    In this section we show the well-posedness of the two derived novel current controlled piezoelectric composite models in the sense of semigroup theory \cite{CurtainZwart1995introduction}. More precisely, we define the associated operators of \eqref{ch3:eq:PDE_currentactuator_EBBT_FD} and \eqref{ch3:eq:PDE_currentactuator_TBT_FD} and show with use of the Lumer-Philips theorem \cite{lumer1961}, that both operators are in fact generators of a strongly continuous semigroup of contractions. 
      \begin{theorem}{\bf (Lumer-Phillips theorem \cite{lumer1961})}\label{ch3:theorem:Lumer-Phillips}
    The closed and densely defined operator $A$ generates a strongly continuous semigroup of contractions $T(t)$ on $X$, if and only if both $A$ and its adjoint $A^*$ are dissipative, i.e.
    \begin{align}\label{ch3:disspative_operator}
    \begin{array}{ccc}
        \angles{A \vect{x},\vect{x}}_X&\leq &0  \quad\text{ for all} \quad x\in\Dom(A), \\
        \angles{A^*{\vect{x}},\vect{x}}_X&\leq &0 \quad \text{ for all}  \quad x\in\Dom(A^\ast).
    \end{array}
    \end{align}
    \qedwhite
\end{theorem}
        Let the length of the beam be $\ell=1$ and define $H^1_0(0,1):=\setof{f\in H^1(0,1)\mid f(0)=0}$, with $H^1(0,1)$ denoting the first order Sobolev space and let $L^2(0,1)$ denote the space of square integrable functions. Firstly, we show the well-posedness of the fully dynamic piezoelectric composite  with \ac{EBBT} \eqref{ch3:eq:PDE_currentactuator_EBBT_FD}. Subsequently, we show the well-posedness of the fully dynamic piezoelectric composite with \ac{TBT} \eqref{ch3:eq:PDE_currentactuator_TBT_FD}.
        
    \subsection{Well-posedness Piezoelectric composite with {EBBT}}\label{ch3:sec:wellposedEBBT}
        Inspired by \eqref{ch3:eq:Hamiltonian_current_FD_EBBT}, define the linear space
        	\begin{align*}
	\mathcal{X}_{EB}=\left\{\vect{x}\in H^1_0(0,1)\times H_0^1(0,1)\times H^1_0(0,1)\right.\\
	\left. L^2(0,1) \times L^2(0,1) \times L^2(0,1) \right\}
	\end{align*}
  	and inner product
	\begin{align}
	\begin{split}
	&\angles{\vect{x},\vect{y}}_\mathcal{X}:=\\
	&\quad \int_{0}^{1}\left[C_Ax_1'y_1'+C_I x_2'y_2'- C_{I_0}\braces{x_1'y_2'+x_2'y_1'}\right. \\
	&\quad+ \left. \frac{1}{\mu}	A^p x_3' z_3' +\rho_A {x}_4{y}_4 +\rho_  I{x}_5{y}_5 \right. \\
		&\quad \left.-\rho_{I_0}\braces{{x}_4{y}_5+{x}_5{y}_4}+\epsilon_{33} A^p {x_6}y_6\right]dz,
	\end{split}
	\end{align}
	where the prime indicate the spatial derivative with respect to $z_1$. The inner product $\angles{. ,.}_{\mathcal{X}_{EB}}$ induces the norm $\norm{\vect{x}}^2_{\mathcal{X}}=\angles{\vect{x},\vect{x}}_{\mathcal{X}_{EB}}=2 \mathcal{H}_{EB}(t)$ on $\mathcal{X}_{EB}$, see \eqref{ch3:eq:Hamiltonian_current_FD_EBBT}. For simplicity, denote the spatial variable $z:=z_1$, additionally let $\partial_{z}:=\frac{\partial}{\partial z}$, and define $\vect{x}:=\begin{bmatrix}v & w_{z} & \Phi & \dot{v} & \dot{w}_{z} & \dot{\Phi}	\end{bmatrix}^T$ to be the state, and current input $u(t)=\mathcal{I}(t)$. Furthermore, define the coefficients 
	\begin{equation}
	    \begin{aligned}
	        	a_{41}&=\frac{\rho_I C_A - \rho_{I_0} C_{I_0}}{\rho_A\rho_I-\rho_{I_0}^2}, \quad &a_{51}&=\frac{\rho_{A} C_{I_0}-\rho_{I_0} C_A  }{\rho_A\rho_I-\rho_{I_0}^2}, \nonumber  \\
            	a_{42}&=\frac{\rho_{I} C_{I_0}-\rho_{I_0} C_I }{\rho_A\rho_I-\rho_{I_0}^2}, \quad 	&a_{52}&=\frac{\rho_{A} C_I - \rho_{I_0} C_{I_0}}{\rho_A\rho_I-\rho_{I_0}^2},   \nonumber\\
            	a_{46}&=\gamma \frac{\rho_I A^p-\rho_{I_0}I^p_0}{\rho_A\rho_I-\rho_{I_0}^2}, \quad  	&a_{56}&=\gamma \frac{\rho_{A}I^p_0-\rho_{I_0} A^p}{\rho_A\rho_I-\rho_{I_0}^2}, \nonumber \\
            	a_{63}&= \frac{1}{\epsilon_{33}\mu},~ a_{64}=\gamma\frac{1}{\epsilon_{33}}, \quad &a_{65}&=\gamma\frac{I^p_0}{\epsilon_{33}A^p},  \label{ch3:eq:coefficientsA_EBBT}
	    \end{aligned}
	\end{equation}
    and denote the $n \times n$ identity operator by $\mathbf{I}_n$.

	Such that the operator corresponding to \eqref{ch3:eq:PDE_currentactuator_EBBT_FD} can be written as follows,	\begin{equation}\label{ch3:eq:opA_FD_EBBT}
    	\begin{gathered}
        	\mathcal{A}_{EB}:\Dom(\mathcal{A}_{EB})\subset \mathcal{X}_{EB}\rightarrow\mathcal{X},\\
            	\mathcal{A}_{EB}=\\
            	\left[
            	\begin{array}{ccc|ccc}
            	& && & \mathbf{I}_3 &\\ \hline
            	a_{41}\partial_{z}^2&-a_{42}\partial_{z}^2 & & & &-a_{46}\partial_{z}\\
            	-a_{51}\partial_{z}^2&a_{52} \partial_{z}^2& & & &a_{56}\partial_{z}\\
            	& & a_{63}\partial_{z}^2 &-a_{64} \partial_{z}&a_{65}\partial_{z}&
            	\end{array}
            	\right]
    	\end{gathered}
	\end{equation}
	with 
	\begin{align}\label{ch3:eq:domain_A_EBBT}
        &\Dom(\mathcal{A}_{EB})= \nonumber \\
        &\quad \left\{ \vect{x}\in \mathcal{X} \mid 	C_A x_1'(1)-C_{I_0}x_2'(1)  -\gamma A^p x_6(1)=0, \right.  \nonumber\\ 
	    &\quad\left. C_I x_2'(1)-C_{I_0}x_1'(1)+\gamma I^p_0  x_6(1)=0,\right.  \nonumber \\ 
	    &\quad\left.\frac{1}{\mu}A^p x_3'(1)=0.\right\}
	\end{align}
    is densely defined in $\mathcal{X}$. Denote the input $u(t)=\mathcal{I}(t)$ and define
    \begin{align}
        B=\begin{bmatrix}  0 & 0 &0&0&0& -\frac{1}{2g_b \varepsilon_{33}} \end{bmatrix}^T,
    \end{align}
    then, the behaviour of a current actuated piezoelectric composite with \ac{EBBT} is described by
      
        \begin{align*}
            \dot{\vect{x}}=\mathcal{A}_{EB}\vect{x}+B u(t).
        \end{align*}
    To establish the well-posedness of the operator \eqref{ch3:eq:opA_FD_EBBT} in the sense of semigroup theory, we set $u(t)= 0$ and make use of the following Lemma.
    
    \begin{lemma} \label{ch3:lem:skewadjoint_A_EBBT}
        The adjoint $\mathcal{A}_{EB}^\ast$ of the operator $\mathcal{A}_{EB}$, defined in \eqref{ch3:eq:opA_FD_EBBT}, is skew-adjoint. More precisely,
        \begin{equation}
            \mathcal{A}_{EB}^\ast=-\mathcal{A}_{EB}, \text{with}
            \Dom(\mathcal{A}_{EB}^\ast)=\Dom(\mathcal{A}_{EB})
        \end{equation}
    \end{lemma}
    \begin{proof}
    For any $\vect{x}=\begin{bmatrix}x_1 & \dots &x_6\end{bmatrix}^T$ and $\vect{y}=\begin{bmatrix}y_1 & \dots &y_6\end{bmatrix}^T$ $\in\Dom(\mathcal{A}_{EB})$ we have,
    \begin{subequations}
    \begin{align}
    \begin{split}
       & \langle \mathcal{A}_{EB} \vect{x}, \vect{v}\rangle_{\mathcal{X}}=
             \int_0^1 \braces{C_A x_4'-C_{I_0}x_5' }y_1'\\
            &\quad +\braces{C_I x_5'-C_{I_0}x_4' }y_2' + \braces{\frac{1}{\mu}A^p x_6'}y_3'\\
            &\quad + \braces{C_A x_1''-C_{I_0}x_2''-\gamma A^p x_6' }y_4\\
            &\quad +\braces{C_I x_2''-C_{I_0}x_1''+\gamma I_0^p x_6' }y_5\\
            &\quad + \left( \frac{1}{\mu}A^p x_3'' -\gamma\left(A^p x_4'-I_0^p x_5'\right)\right)y_6 ~dz  
    \end{split}\\
    \begin{split}
       &=\int_0^1 -\braces{C_A x_4-C_{I_0}x_5 }y_1''\\
            &\quad -\braces{C_I x_5-C_{I_0}x_4 }y_2'' - \braces{\frac{1}{\mu}A^p x_6}y_3''\\
            &\quad - \braces{C_A x_1'-C_{I_0}x_2'-\gamma A^p x_6 }y_4'\\
            &\quad -\braces{C_I x_2'-C_{I_0}x_1'+\gamma I_0^p x_6 }y_5'\\
            &\quad - \left( \frac{1}{\mu}A^p x_3' -\gamma\left(A^p x_4-I_0^p x_5\right)\right)y_6' ~dz \\
            &\quad + \brackets{ \braces{C_A x_4-C_{I_0}x_5 }y_1' }_0^1\\
            &\quad + \brackets{ \braces{C_I x_5-C_{I_0}x_4 }y_2' }_0^1\\ 
            &\quad + \brackets{  \braces{\frac{1}{\mu}A^p x_6}y_3' }_0^1\\ 
            &\quad + \brackets{ \braces{C_A x_1'-C_{I_0}x_2'-\gamma A^p x_6 }y_4 }_0^1\\
            &\quad + \brackets{ \braces{C_I x_2'-C_{I_0}x_1'+\gamma I_0^p x_6 }y_5 }_0^1\\ 
            &\quad + \brackets{ \braces{\frac{1}{\mu}x_3'-\gamma\braces{A^p x_4-I_0^p x_5}  }y_6 }_0^1   
            \end{split}\\
    \begin{split}
            & =   \int_0^1 -\braces{C_A y_4'-C_{I_0}y_5' }x_1'\\
            &\quad -\braces{C_I y_5'-C_{I_0}y_4' }x_2' + \braces{\frac{1}{\mu}A^p y_6'}x_3'\\
            &\quad - \braces{C_A y_1''-C_{I_0}y_2''-\gamma A^p y_6' }x_4\\
            &\quad -\braces{C_I y_2''-C_{I_0}y_1''+\gamma I_0^p y_6' }x_5\\
            &\quad - \left( \frac{1}{\mu}A^p y_3'' -\gamma\left(A^p y_4'-I_0^p y_5'\right)\right)x_6 ~dz \\
            &\quad + \brackets{ \braces{C_A y_1'-C_{I_0}y_2'-\gamma A^p y_6 }x_4 }_0^1\\
            &\quad + \brackets{ \braces{C_I y_2'-C_{I_0}y_1'+\gamma I_0^p y_6 }x_5 }_0^1\\ 
            &\quad + \brackets{ \braces{\frac{1}{\mu}y_3' }x_6 }_0^1\\ 
            &\quad + \brackets{ \braces{C_A x_1'-C_{I_0}x_2'-\gamma A^p x_6 }y_4 }_0^1\\
            &\quad + \brackets{ \braces{C_I x_2'-C_{I_0}x_1'+\gamma I_0^p x_6 }y_5 }_0^1\\ 
            &\quad + \brackets{ \braces{\frac{1}{\mu}x_3'  }y_6 }_0^1   
   \end{split}  \\
   \begin{split}
        & = \langle  \vect{u},-\mathcal{A}_{EB} \vect{v}\rangle_{\mathcal{X}_{EB}} = \langle  \vect{u},\mathcal{A}_{EB}^\ast \vect{v}\rangle_{\mathcal{X}_{EB}},\quad \quad 
        \end{split}
        \end{align}
        \end{subequations}
        
    where we made use of \ac{IBP}, the domain of $\mathcal{A}_{EB}$,  and imposed
    \begin{align} \begin{split}
        y_1(0)=y_2(0)=y_3(0)=0\\
        C_A y_1'(1)-C_{I_0}y_2'(1)  -\gamma A^p y_6(1)=0\\
        C_I y_2'(1)-C_{I_0}y_1'(1)+\gamma I^p_0  y_6(1)=0\\
        \frac{1}{\mu}A^p y_3'(1)=0
            \end{split}
    \end{align}
    for $\mathcal{A}_{EB}^\ast$, similarly to $\mathcal{A}_{EB}$. We have shown that 
    $\mathcal{A}_{EB}^\ast = -\mathcal{A}_{EB}$ and $\Dom(\mathcal{A}_{EB}^\ast) = \Dom(\mathcal{A}_{EB})$.  Hence, we conclude that $\mathcal{A}_{EB}$ is skew-adjoint.
   \end{proof}
    Now we are able to establish the well-posedness of the novel current actuated piezoelectric composite with \ac{EBBT} in the absence of control, i.e. $\mathcal{I}(t)=0$.
    
    \begin{theorem}
        The operator $\mathcal{A}_{EB}$, defined in \eqref{ch3:eq:opA_FD_EBBT}, generates a semigroup of contractions, satisfying $\norm{T(t)}\leq 1$ on $\mathcal{X}_{EB}$.
    \end{theorem}
    \begin{proof}
    The closed and densely defined operator $\mathcal{A}$ satisfies    
    \begin{align} \label{ch3:eq:Adissipative_EBBT}
    \begin{split}
         &\langle\mathcal{A}_{EB}\vect{x},\vect{x}\rangle= \int_0^1 \braces{C_A x_4'-C_{I_0}x_5' }x_1'\\
            &\quad +\braces{C_I x_5'-C_{I_0}x_4' }x_2' + \braces{\frac{1}{\mu}A^p x_6'}x_3'\\
            &\quad + \braces{C_A x_1''-C_{I_0}x_2''-\gamma A^p {x}_6' }x_4\\
            &\quad +\braces{C_I x_2''-C_{I_0}x_1''+\gamma I_0^p {x}_6' }x_5\\
            &\quad + \left( \frac{1}{\mu}A^p x_3'' -\gamma\left(A^p x_4'-I_0^p x_5'\right)\right)x_6 ~dz  
    \end{split}\\
       \begin{split}
            & =\int_0^1 \braces{C_A x_4'-C_{I_0}x_5' }x_1'\\
            &\quad +\braces{C_I x_5'-C_{I_0}x_4' }x_2' - \braces{\frac{1}{\mu}A^p x_6}x_3''\\
            &\quad - \braces{C_A x_1'-C_{I_0}x_2'-\gamma A^p {x}_6 }x_4'\\
            &\quad -\braces{C_I x_2'-C_{I_0}x_1'+\gamma I_0^p {x}_6 }x_5'\\
            &\quad + \left( \frac{1}{\mu}A^p x_3'' -\gamma\left(A^p x_4'-I_0^p x_5'\right)\right)x_6 ~dz  \\
            &\quad + \brackets{ \braces{\frac{1}{\mu}x_3' }x_6 }_0^1\\ 
            &\quad + \brackets{ \braces{C_A x_1'-C_{I_0}x_2'-\gamma A^p {x}_6 }x_4 }_0^1\\
            &\quad + \brackets{ \braces{C_I x_2'-C_{I_0}x_1'+\gamma I_0^p {x}_6 }x_5 }_0^1\\ 
            &=0
       \end{split}
    \end{align}
    where we made of \ac{IBP} and the domain of $\mathcal{A}_{EB}$. Furthermore, for the adjoint of $\mathcal{A}_{EB}$, we have that
    \begin{align}\label{ch3:eq:A_ast_dissipative_EBBT}
        \langle \mathcal{A}_{EB}^\ast \vect{x},\vect{x}\rangle =  \langle \mathcal{-A}_{EB} \vect{x},\vect{x}\rangle = - \langle \mathcal{A}_{EB} \vect{x},\vect{x}\rangle = 0,
    \end{align}
    where we made use of Lemma \ref{ch3:lem:skewadjoint_A_EBBT} and \eqref{ch3:eq:Adissipative_EBBT}.\\
    This shows that both $\mathcal{A}_{EB}$ and $\mathcal{A}_{EB}^\ast$ satisfy the dissipative properties of a semigroup of contractions. Therefore, with the use of Theorem \ref{ch3:theorem:Lumer-Phillips}, we conclude that the operator \eqref{ch3:eq:opA_FD_EBBT} is a generator of a strongly continuous semigroup and is well-posed.
    \end{proof}
    
     \subsection{Well-posedness Piezoelectric composite with {TBT}}\label{ch3:sec:wellposedTBT}
        We follow the same procedure as in Section \ref{ch3:sec:wellposedEBBT}, but now for the piezoelectric composite with \ac{TBT}. Inspired by \eqref{ch3:eq:Hamiltonian_current_FD_TBT}, define the linear space
        	\begin{align*}
	\mathcal{X}_{T}=\left\{\vect{x}\in H^1_0(0,1)\times H^1_0(0,1)\times H_0^1(0,1)\times H^1_0(0,1)\right.\\
	\left. \times L^2(0,1) \times L^2(0,1) \times L^2(0,1) \times L^2(0,1) \right\}
	\end{align*}
  	and inner product
	\begin{align}
	\begin{split}
	&\angles{\vect{x},\vect{y}}_{\mathcal{X}_{T}}:=\\
	&\quad \int_{0}^{1}\left[C_Ax_1'y_1'+C_I x_2'y_2'- C_{I_0}\braces{x_1'y_2'+x_2'y_1'}\right. \\
	&\quad+ \left.  \tfrac{1}{2}G_A (x_3' y_3' + x_2 y_2 -(x_3'y_2 +x_2 y_3') ) \right. \\
	&\quad+ \left.  \frac{1}{\mu}	A^p x_4' z_4' +\rho_A {x}_5{y}_5 + {x}_7{y}_7+\rho_ I{x}_6{y}_6 \right. \\
		&\quad \left.-\rho_{I_0}\braces{{x}_5{y}_6+{x}_6{y}_5} + \rho_A{x}_7{y}_7+\epsilon_{33} A^p {x_8}y_8\right]dz,
	\end{split}
	\end{align}
	where the prime indicate the spatial derivative with respect to $z_1$. The inner product $\angles{. ,.}_\mathcal{X}$ induces the norm $\norm{\vect{x}}^2_{\mathcal{X}_{T}}=\angles{\vect{x},\vect{x}}_{\mathcal{X}_{T}}=2 \mathcal{H}_T(t)$ on $\mathcal{X}_{T}$, see \eqref{ch3:eq:Hamiltonian_current_FD_TBT}. For simplicity, denote the spatial variable $z:=z_1$, additionally let $\partial_{z}:=\frac{\partial}{\partial z}$, and define $\vect{x}:=\begin{bmatrix}v & \phi &  w & \Phi & \dot{v} & \dot{\phi} & \dot{w} & \dot{\Phi}	\end{bmatrix}^T$ to be the state. Furthermore, define the coefficients 
	\begin{equation}
	    \begin{aligned}
	        	a_{51}&=\frac{\rho_I C_A - \rho_{I_0} C_{I_0}}{\rho_A\rho_I-\rho_{I_0}^2}, \quad &a_{61}&=\frac{\rho_{A} C_{I_0}-\rho_{I_0} C_A  }{\rho_A\rho_I-\rho_{I_0}^2}, \nonumber  \\
            	a_{52}&=\frac{\rho_{I} C_{I_0}-\rho_{I_0} C_I }{\rho_A\rho_I-\rho_{I_0}^2}, \quad 	&a_{62}&=\frac{\rho_{A} C_I - \rho_{I_0} C_{I_0}}{\rho_A\rho_I-\rho_{I_0}^2},   \nonumber\\
            	a_{53}&= \frac{\rho_{I_0}G_A}{2(\rho_A\rho_I-\rho_{I_0}^2)},\quad &a_{63}&= \frac{\rho_{A}G_A}{2(\rho_A\rho_I-\rho_{I_0}^2)} \nonumber\\
            	a_{58}&=\gamma \frac{\rho_I A^p-\rho_{I_0}I^p_0}{\rho_A\rho_I-\rho_{I_0}^2}, \quad  	&a_{68}&=\gamma \frac{\rho_{A}I^p_0-\rho_{I_0} A^p}{\rho_A\rho_I-\rho_{I_0}^2}, \nonumber \\
            	a_{7}&=\frac{G_A}{2 \rho_A}    , \quad &a_{83}&= \frac{1}{\epsilon_{33}\mu},    \nonumber\\
            a_{84}&=\gamma\frac{1}{\epsilon_{33}}, \quad &a_{85}&=\gamma\frac{I^p_0}{\epsilon_{33}A^p}  \label{ch3:eq:coefficientsA_TBT}
	    \end{aligned}
	\end{equation}

	Such that the operator corresponding to \eqref{ch3:eq:PDE_currentactuator_TBT_FD} can be written as follows,	
	\small{
	\begin{equation}\label{ch3:eq:opA_FD_TBT}
    	\begin{gathered}
        	\mathcal{A}_{T}:\Dom(\mathcal{A}_{T})\subset \mathcal{X}_{T}\rightarrow\mathcal{X}_{T},\\
            	\mathcal{A}_{T}=\left[\begin{array}{cc}
            	      & \mathbf{I}_4 \\
            	    \bar{A}_a & \bar{A}_b
            	\end{array} \right]\\
    	\end{gathered}
	\end{equation}}
	with 
    	\begin{align}
    	   {A}_a&=\left[\begin{array}{cccc}
    	        a_{51}\partial_{z}^2 & -a_{52a}\partial_{z}^2- a_{53}  &  a_{53}\partial_{z} & \\
                	-a_{61}\partial_{z}^2 & a_{62a}\partial_{z}^2- a_{63} & a_{63}\partial_{z}&  \\
                	& -a_7 \partial_{z} & a_7 \partial_{z}^2 &   \\
                	& & & a_{84}\partial_{z}^2
    	    \end{array} \right], \\
        	 {A}_b&=\left[\begin{array}{cccc}
        	    & & &-a_{58}\partial_{z}\\
        	    & & &a_{68}\partial_{z} \\
        	     & &   \\
        	     -a_{85} \partial_{z}&a_{86}\partial_{z}&
        	   \end{array} \right],
        \end{align}
	and
	
	\begin{align}\label{ch3:eq:domain_A_TBT}
        &\Dom(\mathcal{A}_{T})= \nonumber \\
        &\quad \left\{ \vect{x}\in \mathcal{X} \mid 	C_A x_1'(1)-C_{I_0}x_2'(1)  -\gamma A^p x_6(1)=0, \right.  \nonumber\\ 
	    &\quad\left. C_I x_2'(1)-C_{I_0}x_1'(1)+\gamma I^p_0  x_6(1) =0,\right.  \nonumber \\ 
	    &\quad\left. \frac{G_A}{2}(x_3'(1)-x_2(1)) ,\quad \frac{1}{\mu}A^p x_3'(1)=0.\right\}
	\end{align}
 
    is densely defined in $\mathcal{X}_{T}$. Denote the input $u(t)=\mathcal{I}(t)$ and define
    \begin{align}
        B_{T}=\begin{bmatrix}  0 & 0 &0&0&0&0&0 -\frac{1}{2g_b \varepsilon_{33}} \end{bmatrix}^T,
    \end{align}
    then, the behaviour of a current actuated piezoelectric composite with \ac{TBT} is described by
     describes 
        \begin{align*}
            \dot{\vect{x}}=\mathcal{A}_{T}\vect{x}+B_{T} u(t).
        \end{align*}
    To establish the well-posedness of the operator \eqref{ch3:eq:opA_FD_EBBT} in the sense of semigroup theory, we set $u(t)= 0$ and make use of the following Lemma.
    
     \begin{lemma} \label{ch3:lem:skewadjoint_A_TBT}
         The adjoint $\mathcal{A}_{T}^\ast$ of the operator $\mathcal{A}_{T}$, defined in \eqref{ch3:eq:opA_FD_TBT}, is skew-adjoint.
    \end{lemma}
    \begin{proof}
    For any $\vect{u}=\begin{bmatrix}x_1 & \dots &x_8\end{bmatrix}^T$ and $\vect{y}=\begin{bmatrix}y_1 & \dots &y_8\end{bmatrix}^T$ $\in\Dom(\mathcal{A}_{T})$ we have,
    \begin{subequations}
    \begin{align}
    \begin{split}
       & \langle \mathcal{A}_{T} \vect{u}, \vect{v}\rangle_{\mathcal{X}}=
             \int_0^1 \braces{C_A x_5'-C_{I_0}x_6' }y_1'\\
            &\quad +\braces{C_I x_6'-C_{I_0}x_5' }y_2' +\frac{1}{2}G_A(x_7'-x_6)y_3'\\
            &\quad  -\frac{1}{2}G_A(x_7'-x_6)y_2+ \braces{\frac{1}{\mu}A^p x_8'}y_4'\\
            &\quad + \braces{C_A x_1''-C_{I_0}x_2''-\gamma A^p x_8' }y_5\\
            &\quad +\braces{C_I x_2''-C_{I_0}x_1''+\frac{1}{2}G_A(x_3'-x_2)+\gamma I_0^p x_8' }y_6\\
            &\quad + \frac{1}{2}G_A(x_3''-x_2')y_7\\
            &\quad + \left( \frac{1}{\mu}A^p x_4'' -\gamma\left(A^p x_5'-I_0^p x_6'\right)\right)y_8 ~dz  
    \end{split}\\
    \begin{split}
       &=\int_0^1 -\braces{C_A x_5-C_{I_0}x_6 }y_1''\\
            &\quad -\braces{C_I x_6-C_{I_0}x_5 }y_2'' -\frac{1}{2}G_A x_7y_3'' - \frac{1}{2}G_A x_6y_3' \\
            &\quad  +\frac{1}{2}G_Ax_7y_2' +\frac{1}{2}G_A x_6y_2- \braces{\frac{1}{\mu}A^p x_8}y_4''\\
            &\quad - \braces{C_A x_1'-C_{I_0}x_2'-\gamma A^p x_8 }y_5'\\  
            &\quad -\braces{C_I x_2'-C_{I_0}x_1'+\gamma I_0^p x_8 }y_6'+\frac{1}{2}G_Ax_3'y_6\\
            &\quad - \frac{1}{2}G_A x_2y_6 - \frac{1}{2}G_A(x_3'-x_2)y_7'\\
            &\quad - \left( \frac{1}{\mu}A^p x_4' -\gamma\left(A^p x_5-I_0^p x_6\right)\right)y_8' ~dz  \\
            &\quad + \brackets{ \braces{C_A x_5-C_{I_0}x_6 }y_1' }_0^1 + \brackets{ \braces{C_I x_6-C_{I_0}x_5 }y_2' }_0^1\\
            &\quad + \brackets{\frac{1}{2}G_A x_7y_3'-\frac{1}{2}G_A x_7y_2}_0^1 +  \brackets{  {\frac{1}{\mu}A^p x_6}y_3' }_0^1 \\
            &\quad + \brackets{ \braces{C_A x_1'-C_{I_0}x_2'-\gamma A^p x_8 }y_5 }_0^1\\
            &\quad + \brackets{ \braces{C_I x_2'-C_{I_0}x_1'+\gamma I_0^p x_8 }y_6 }_0^1\\ 
            &\quad + \brackets{\frac{1}{2}G_A(x_3'-x_2)y_7}_0^1\\ 
            &\quad + \brackets{ \braces{\frac{1}{\mu}x_4'-\gamma\braces{A^p x_5-I_0^p x_5}  }y_8 }_0^1   
            \end{split}\\
    \begin{split}
            & =   -\braces{C_A y_5'-C_{I_0}y_6' }x_1'\\
            &\quad -\braces{C_I y_6'-C_{I_0}y_5' }x_2' -\frac{1}{2}G_A(y_7'-y_6)x_3'\\
            &\quad  +\frac{1}{2}G_A(y_7'-y_6)x_2- \braces{\frac{1}{\mu}A^p y_8'}x_4'\\
            &\quad - \braces{C_A y_1''-C_{I_0}y_2''-\gamma A^p y_8' }x_5\\
            &\quad -\braces{C_I y_2''-C_{I_0}y_1''+\frac{1}{2}G_A(y_3'-y_2)+\gamma I_0^p y_8' }x_6\\
            &\quad - \frac{1}{2}G_A(y_3''-y_2')x_7\\
            &\quad - \left( \frac{1}{\mu}A^p y_4'' -\gamma\left(A^p y_5'-I_0^p y_6'\right)\right)x_8 ~dz \\ 
            &\quad + \brackets{ \braces{C_A y_1'-C_{I_0}y_2'-\gamma A^p y_8 }x_5 }_0^1\\
            &\quad + \brackets{ \braces{C_I y_2'-C_{I_0}y_1'+\gamma I_0^p y_8 }x_6 }_0^1\\ 
            &\quad + \brackets{\frac{1}{2}G_A(y_3'-y_2)x_7}_0^1  + \brackets{ {\frac{1}{\mu}y_4' }x_8 }_0^1\\ 
            &\quad + \brackets{ \braces{C_A x_1'-C_{I_0}x_2'-\gamma A^p x_8 }y_5 }_0^1\\
            &\quad + \brackets{ \braces{C_I x_2'-C_{I_0}x_1'+\gamma I_0^p x_8 }y_6 }_0^1\\ 
            &\quad + \brackets{\frac{1}{2}G_A(x_3'-x_2)y_7}_0^1 + \brackets{ {\frac{1}{\mu}x_4'  }y_8 }_0^1   
   \end{split}  \\
   \begin{split}
        & = \langle  \vect{u},-\mathcal{A}_{T} \vect{v}\rangle_{\mathcal{X}_{T}} = \langle  \vect{u},\mathcal{A}_{T}^\ast \vect{v}\rangle_{\mathcal{X}_{T}},\quad \quad 
        \end{split}
        \end{align}
        \end{subequations}
    where we made use of \ac{IBP}, the domain of $\mathcal{A}_{T}$,  and imposed
    \begin{align} \begin{split}
        y_1(0)=y_2(0)=y_3(0)=v(4)=0\\
        C_A y_1'(1)-C_{I_0}y_2'(1)  -\gamma A^p y_8(1)=0\\
        C_I y_2'(1)-C_{I_0}y_1'(1)+\gamma I^p_0  y_8(1)=0\\
        \frac{1}{2}G_A(y_3'(1)-y_2(1))=0, \quad \frac{1}{\mu}A^p y_4'(1)=0
            \end{split}
    \end{align}
    for $\mathcal{A}_{T}^\ast$, similarly to $\mathcal{A}_{T}$. We have shown that $\mathcal{A}^\ast = -\mathcal{A}_{T}$ and, furthermore, $\Dom(\mathcal{A}_{T}^\ast) = \Dom(\mathcal{A}_{T})$. Hence, we conclude that $\mathcal{A}_{T}^\ast$ is skwe-adjoint.
   \end{proof}
   
    Now we are able to establish the well-posedness of the novel current actuated piezoelectric composite with \ac{TBT} in the absence of control, i.e. $\mathcal{I}(t)=0$.
    
     \begin{theorem}
        The operator $\mathcal{A}_{T}$, defined in \eqref{ch3:eq:opA_FD_TBT}, generates a strongly continuous semigroup of contractions, satisfying $\norm{T(t)}\leq 1$ on $\mathcal{X}_{T}$.
    \end{theorem}
    \begin{proof}
    The closed and desnely defined operator $\mathcal{A}_{T}$ satisfies    
    \begin{align} \label{ch3:eq:Adissipative_TBT}
    \begin{split}
         &\langle\mathcal{A}_{T}\vect{x},\vect{x}\rangle= \int_0^1 \braces{C_A x_5'-C_{I_0}x_6' }x_1'\\
            &\quad +\braces{C_I x_6'-C_{I_0}x_5' }x_2' + \frac{1}{2}G_A(x_7'-x_6)x_3'\\
           &\quad  -\frac{1}{2}G_A(x_7'-x_6)x_2 + {\frac{1}{\mu}A^p x_8'}x_4'\\
            &\quad + \braces{C_A x_1''-C_{I_0}x_2''-\gamma A^p {x}_8' }x_5\\
            &\quad +\braces{C_I x_2''-C_{I_0}x_1''+\gamma I_0^p {x}_8' }x_6\\
            &\quad + \frac{1}{2}G_A(x_3'-x_2)x_6 + \frac{1}{2}G_A(x_3''-x_2')x_7 \\
            &\quad + \left( \frac{1}{\mu}A^p x_4'' -\gamma\left(A^p x_5'-I_0^p x_6'\right)\right)x_8 ~dz  
    \end{split}\\
       \begin{split}
            & =\int_0^1 \braces{C_A x_5'-C_{I_0}x_6' }x_1'\\
            &\quad +\braces{C_I x_6'-C_{I_0}x_5' }x_2' + \frac{1}{2}G_A(x_7'-x_6)x_3'\\
           &\quad  -\frac{1}{2}G_A(x_7'-x_6)x_2 - {\frac{1}{\mu}A^p x_8}x_4''\\
             &\quad - \braces{C_A x_1'-C_{I_0}x_2'-\gamma A^p {x}_8 }x_5'\\
            &\quad -\braces{C_I x_2'-C_{I_0}x_1'+\gamma I_0^p {x}_8 }x_6'\\
            &\quad +\frac{1}{2}G_A(x_3'-x_2)x_6 - \frac{1}{2}G_A(x_3'-x_2)x_7' \\
            &\quad + \left( \frac{1}{\mu}A^p x_4'' -\gamma\left(A^p x_5'-I_0^p x_6'\right)\right)x_8 ~dz \\
            &\quad + \brackets{ {\frac{1}{\mu}x_8 }x_4' }_0^1\\ 
            &\quad + \brackets{ \braces{C_A x_1'-C_{I_0}x_2'-\gamma A^p {x}_8 }x_5 }_0^1\\
            &\quad + \brackets{ \braces{C_I x_2'-C_{I_0}x_1'+\gamma I_0^p {x}_8 }x_6 }_0^1\\ 
            &\quad +\brackets{\frac{1}{2}G_A(x_3'-x_2)x_7}_0^1 \\
            &=0
       \end{split}
    \end{align}
    where we made of \ac{IBP} and the domain of $\mathcal{A}_{T}$. Furthermore, for the adjoint of $\mathcal{A}_{T}$, we have that
    \begin{align}\label{ch3:eq:A_ast_dissipative_TBT}
        \langle \mathcal{A}_{T}^\ast \vect{x},\vect{x}\rangle =  \langle \mathcal{-A}_{T} \vect{x},\vect{x}\rangle = - \langle \mathcal{A}_{T} \vect{x},\vect{x}\rangle = 0,
    \end{align}
    where we made use of Lemma \ref{ch3:lem:skewadjoint_A_TBT} and \eqref{ch3:eq:Adissipative_TBT}.\\
    This shows that both $\mathcal{A}_{T}$ and $\mathcal{A}_{T}^\ast$ satisfy the dissipative properties of a strongly continuous semigroup of contractions. Therefore, with the use of Theorem \ref{ch3:theorem:Lumer-Phillips}, we conclude that the operator \eqref{ch3:eq:opA_FD_TBT} is well-posed.
    \end{proof}

\section{Comparison with other current-controlled piezoelectric beam and composite models}\label{ch3:sec:comparison_beamandcomp}


        
        The approach for the novel current piezoelectric models presented here leads to the definition of the magnetic flux density and requires a combined Lagrangian to include the (co-)energies of the well-posed systems. More precisely, for the electromagnetic part, we couple the co-Lagrangian to the Lagrangian of the mechanical part through the non-energetic coupling terms known as traditors \cite{Duinkers1959Traditors}. \blue{These coupling terms are present in the combined Lagrangian, however, not in the Hamiltonian.} This approach is inspired by the treatment of Maxwell's equations for voltage-actuated piezoelectric beams, where the charge is defined by integrating the electric displacement, see for instance \cite{MenOSIAM2014}. Similarly, we obtain an expression for the magnetic flux by integrating the total magnetic field passing through the area. This approach deviates from two existing approaches to derive current actuated piezoelectric beams and composite models, i.e. the \textit{current-through-the-boundary} and the use of magnetic vector potentials. \\
        {These two other principles -to derive current actuated piezoelectric beams and composite are either from taking a charge actuated piezoelectric beam and adding a dynamical equation on the boundary \cite{OverviewpaperdeJong_arxiv} (Remark 2) or by utilizing the magnetic vector potentials $\Tilde{\vect{A}}$. The first approach results in a current-through-the-boundary type model by mathematically adding an integrator for the charge $Q$ on the boundary, i.e. $\mathcal{I}(t)=\tfrac{d}{dt}Q(t)$, see for instance \cite{OzerTAC2019} for a classical system description and \cite{VenSSIAM2014} for a \ac{pH} system description. In \cite{VenSSIAM2014}, the fully dynamic current actuated system exploits the \ac{pH} formalism utilizing the boundary ports, mimicking the charge integrator without incorporating the equation on the boundary. Physically, either case corresponds to the use of some electric circuitry. The charge and subsequently resulting current actuated systems have similar stabilizability properties, besides that current-through-the boundary type systems, which can utmost asymptotically stabilize the system due to the bounded input operator, see \cite{OzerTAC2019}.}\\        
        The approach using magnetic vector potentials $\Tilde{\vect{A}}$ is described in \cite{MenOCDC2014,OzerTAC2019}. As per Gauss's magnetic law \eqref{ch3:eq:GausssMagnetic}, there exist magnetic vector potentials such that
        \begin{align}\label{ch3:eq:magvectpot}
	        \vect{B}=\nabla\times \Tilde{\vect{A}}.
	    \end{align}
        Substituting \eqref{ch3:eq:magvectpot} into Faraday's law \eqref{ch3:eq:Faraday'sLaw} results in an expression of the electric field $\vect{E}$ with electric scalar potential $\varphi$ for a piezoelectric actuator as follows,
    	\begin{align}\label{ch3:eq:Efield}
    	    \vect{E}&=-\nabla \varphi-\frac{\partial }{\partial t}\tilde{\vect{A}},
    	\end{align}
    	and result in Max-Amperes law expressed in magnetic vector and scalar potentials, as follows
    	\begin{align}\label{ch3:eq:maxamp_magneticvec&salar}
            \mu^{-1}(\nabla (\nabla \cdot \tilde{\vect{A}})-\nabla^2  \tilde{\vect{A}}) = e^T \frac{\partial \epsilon_{11}}{\partial t} -\varepsilon \frac{\partial^2  \tilde{\vect{A}}}{\partial t^2}-\varepsilon \nabla \frac{\partial \varphi}{\partial t}+\vect{J}.	    
    	\end{align}
    	The electric scalar potential $\varphi$ and magnetic vector potentials $\tilde{\vect{A}}$ are not uniquely defined \cite{eom2013maxwell,MenOCDC2014}. Therefore, a gauge function, such as the Coulomb gauge $(\nabla\cdot \tilde{\vect{A}}=0)$ or the Lorentz gauge ($\nabla\cdot \tilde{\vect{A}} + {c}^{-2}\tfrac{\partial \varphi}{\partial t}=0$ is required to uniquely define $\varphi$ and $\tilde{\vect{A}}$ and solve \eqref{ch3:eq:maxamp_magneticvec&salar}, see for instance \cite{eom2013maxwell}. In the case of the Coulomb gauge, we obtain the equation for piezoelectric actuators as follows
    	\begin{align}\label{ch3:eq:maxamp_coulomb_scalar}
    	    \varepsilon_{33}\frac{\partial^2 \tilde{A}}{\partial t^2}&=\frac{1}{mu}\frac{\partial^2 \tilde{A}_3}{\partial z_1^2}+\gamma \frac{\partial \epsilon_{11}}{\partial t} -\varepsilon_{33} \frac{\partial^2 \varphi}{\partial z_3 t}+J_3,
    	\end{align}
    	where an elliptical needs to be solved for $\varphi$, see for instance \cite{OzerTAC2019,MenOMTNS2014,MenOCDC2014}. In the case of the Lorentz gauge, we obtain the equation for piezoelectric actuators
    	\begin{align}\label{ch3:eq:maxamp_lorentz_scalar}
    	    \varepsilon_{33}\frac{\partial^2 \tilde{A}_3}{\partial t^2}&=\frac{1}{mu}\frac{\partial^2 \tilde{A}_3}{\partial z_1^2}+\gamma \frac{\partial \epsilon_{11}}{\partial t} +J_3,
    	\end{align}
        	see for instance \cite{Ozer2020}. Although these gauge conditions do not influence the electric field $E$ and magnetic field $B$, they do have their specific characteristics \cite{eom2013maxwell} and influence the dynamic equations governing the piezoelectric actuator, see for instance \cite{MenOCDC2014,OzerTAC2019}. In \cite{OzerTAC2019}, it has been shown that the purely current actuated fully dynamic piezoelectric system, with two current sources, using electric vector potentials and a Coulomb gauge condition lack the stabilizability property. \blue{Here we would like to point out that the modelling approach proposed in this work circumvents the need for a gauge condition to derive a well-posed system.}
        	
            By comparison of \eqref{ch3:eq:maxamp_lorentz_scalar} and \eqref{ch3:eq:maxamp_novel_scalar}, we see that the method presented in this work, using the magnetic flux \eqref{ch3:eq:magneticflux_flow3}, shows similarities with the Lorentz gauge approach, i.e. $\Phi=-\tilde{A}_3$. However, the definition of the magnetic flux \eqref{ch3:eq:magneticflux_flow3} circumvents the necessity of the gauge condition and is analogous to the derivation of voltage-actuated piezoelectric actuators; see, for instance, \cite{MenOSIAM2014}. Furthermore, the variable $\Phi$ provides a clear physical interpretation \eqref{ch3:eq:magneticflux_flow3}. Therefore, the framework for modelling piezoelectric actuators, beams, and composites can be appended with the modelling approach presented here.
            
            In \cite{VenSSIAM2014}, a non-linear quasi-static current-controlled actuator and composite model is presented. It can be shown that the fully dynamic current-controlled actuator and composite presented in this work and the non-linear current-controlled quasi-static piezoelectric actuator/composite presented in \cite{VenSSIAM2014} are of the same nature. More precisely, by reducing the electromagnetic assumption \eqref{ch3:eq:PDE_currentactuator_TBT_FD} to the quasi-static situation, by letting $\frac{\partial}{\partial z}B_2(=\Phi_{z_1z_1})\rightarrow 0$  and linearize of the non-linear Timoshenko beam theory of the current actuated quasi-static piezoelectric system presented in \cite{VenSSIAM2014}, result in coinciding systems. Therefore, we conclude that the fully dynamic electromagnetic model derived in this work with the use of the magnetic flux \eqref{ch3:eq:magneticflux_flow3} yields a more extensive model by including the electromagnetic coupling than the quasi-static model in \cite{voss2011CDC}. Furthermore, in \cite{VenSSIAM2014}, it has been shown that the quasi-static current actuated model is not stabilizable. The stabilizability for the models \eqref{ch3:eq:PDE_currentactuator_EBBT_FD} and \eqref{ch3:eq:PDE_currentactuator_TBT_FD} is yet unanswered and will be addressed in the next section.\\
        

    
    \section{Feedback stabilization of piezoelectric composite}\label{ch3:sec:Stabz}
    
            To stabilize the fully dynamic current-controlled piezoelectric composite models \eqref{ch3:eq:PDE_currentactuator_EBBT_FD} and \eqref{ch3:eq:PDE_currentactuator_TBT_FD} we investigate a Lyapunov-based control strategy and make use of the following theorem for infinite dimensional systems.
                \begin{theorem}[LaSalle's Invariance Principle \cite{Luo1999SSIDSA}]\label{ch3:thm:LaSalleinv}
                    Let $\mathcal{V}$ be a continuous Lyapunov function for the strongly continuous semigroup $T(t)$ on $X$ and let the largest invariant subset be denoted as
                        $$\mathcal{W}:=\{\vect{x} \mid \tfrac{d}{dt}\mathcal{V}(\vect{x})=0\}. $$
                    If $x\in X$ and the orbits 
                        $$\tilde{\lambda}(\vect{x}) = \bigcup_{t \geq 0} T(t)\vect{x} $$
                     is precompact, then for the distance $d(.,.)$, we have that
                        $$\lim_{t\rightarrow\infty} d(T(t)\vect{x},\mathcal{W})=0.$$
                    Here, by invariance of $\mathcal{W}$ under $T(t)$, we mean $T(t)\mathcal{W}=\mathcal{W}$ for all $t\geq 0$. 
                    \qedwhite
                \end{theorem}
            Recall the corresponding energy functional $\mathcal{H}(t)$ in \eqref{ch3:eq:Hamiltonian_current_FD_EBBT} and \eqref{ch3:eq:Hamiltonian_current_FD_TBT}, where the change of these energies along the trajectories of \eqref{ch3:eq:PDE_currentactuator_EBBT_FD} as follows,
            \begin{align}\label{ch3:eq:Lyap_roc}
                \frac{d}{dt}\mathcal{H}(t)=-\int_0 ^1 (h_2-h_1)\dot{\Phi}(z) \mathcal{I}(t) ~dz = -\int_0^1 \kappa \dot{\Phi}^2(z)  ~dz
            \end{align}
            for control choice 
            \begin{align}\label{ch3:eq:control_choice}
                \mathcal{I}(t)=\tfrac{\kappa}{(h_2-h_1)} \dot{\Phi}(z), ~ \text{with }  \kappa >0 .
            \end{align}
             Then, $\mathcal{H}(t)$ can be considered as a Lyapunov candidate function to prove the asymptotic stability of the closed system. \\
            
        
            For the controlled piezoelectric composite with \ac{EBBT}, define $\vect{x}=\col(y_z,w_{zz},\Phi_z,\rho_A \dot{v}-\rho_{I_0}\dot{w}_z,\rho_I \dot{w}_z - \rho_{I_0}\dot{v}, \tfrac{A^p }{\varepsilon} \dot{\Phi}+\gamma (A^p v_z- I_0 w_{zz}))$ such that $x\in \mathcal{X}_\mathcal{H}$ ($L^2$-space) and denote the closed-loop system obtained via the control choice \eqref{ch3:eq:control_choice} as follows,
            {\begin{align}\label{ch3:eq:closedloop_sys}
                \dot{\vect{x}}&=\mathcal{A}\vect{x}, t>0 \\
                \vect{x}(0)&=\vect{x}_0 \in \Dom(\mathcal{A}) \nonumber
            \end{align}}
            with the closed and densely defined operator
           { \begin{align}\label{ch3:eq:closedloopgen_EBBT}
                    \mathcal{A}\vect{x}&=\begin{pmatrix} \tfrac{\rho_I}{\rho_A\rho_I-\rho_{I_0}^2}x_4' + \tfrac{\rho_{I_0}}{\rho_A\rho_I-\rho_{I_0}^2} x_5'\\ \tfrac{\rho_A}{\rho_A\rho_I-\rho_{I_0}^2} x_5' + \tfrac{\rho_{I_0}}{\rho_A\rho_I-\rho_{I_0}^2}x_4'    \\  \tfrac{1}{\varepsilon A^p} x_6' - \tfrac{\gamma}{\varepsilon}x_1'+\tfrac{\gamma I_0}{\varepsilon A^p}x_2'  \\
                    (C_A+\tfrac{\gamma^2 A^p}{\varepsilon})x_1'-(C_{I_0}+\tfrac{\gamma^2 I_0}{\varepsilon}))x_2'-\frac{\gamma}{\varepsilon} A_p x_6'\\
                    (C_I+\tfrac{\gamma^2 I_0^2}{\varepsilon A^p}))x_2'-(C_{I_0}+\tfrac{\gamma^2 I_0}{\varepsilon}))x_1'+\frac{\gamma}{\varepsilon} I_0 x_6'\\
                    \tfrac{A_p}{\varepsilon} x_3' -\kappa (\tfrac{1}{\varepsilon A^p} x_6 - \tfrac{\gamma}{\varepsilon}x_1+\tfrac{\gamma I_0}{\varepsilon A^p}x_2) \end{pmatrix}\\
                    &=\begin{pmatrix}
                        \dot{v}_z \\ \dot{w}_{zz} \\ \dot{\Phi}_z \\
                        C_A v_{zz} - C_{I_0} w_{zzz}-\gamma A^p \dot{\Phi}_z\\
                        C_I w_{zzz} - C_{I_0} v_{zz} + \gamma I_0 \dot{\Phi}_z\\
                        \frac{A^p}{\mu} q_{zz}
                    \end{pmatrix}
            \end{align}}
            on the domain 
           { \begin{align}
                \begin{split}
                        &\Dom(\mathcal{A})=\{ \vect{x} \in \mathcal{X}_\mathcal{H} |  \\
                        & \quad C_A v_z(1)-C_{I_0}w_{zz}(1)-\frac{\gamma}{\varepsilon} A^p \dot{\Phi}(1)=0,\\
                        & \quad C_Iw_{zz}(1)-C_{I_0}v_z(1)+\frac{\gamma}{\varepsilon} I_0 \dot{\Phi}(1)=0,\\ 
                       & \quad \frac{A^p }{\varepsilon} \Phi_z(1)=0 \}.
               \end{split}
            \end{align}}
           It is straightforward to show that the operator \eqref{ch3:eq:closedloopgen_EBBT} generates a strongly continuous semigroup of contractions $T_{cl}(t)$, which is well-posed in a similar fashion as is shown for {\eqref{ch3:eq:opA_FD_EBBT}}. We are able to prove the asymptotic stability of the closed-loop system. Therefore, define the inner-product $\langle \vect{x},\vect{x}\rangle_\mathcal{H}=2\mathcal{H}_{EB}(t)$ on $X_\mathcal{H}$ and assume the following inequalities on the system parameters.
    \begin{assumption}{System parameter inequalities}\label{ch3:ass:sys_coef}
                \begin{align*}
                    \begin{array}{rlrl}
                       \rho_A \rho_I &\neq \rho_{I_0}^2, \quad\quad   & C_A C_I &\neq C_{I_0}^2
                    \end{array}
                \end{align*}
           \end{assumption}
          
               \begin{theorem}\label{ch3:thm:Asympstabz_EBBT_FD}
                    Let the inequalities in Assumption \ref{ch3:ass:sys_coef} hold and consider the closed-loop system  \eqref{ch3:eq:closedloopgen_EBBT}. Furthermore, consider the Lyapunov candidate functions $\mathcal{V}=\mathcal{H}_{EB}(t)$. Then, the closed-loop system \eqref{ch3:eq:closedloop_sys} is well-posed and asymptotically stable.
                \end{theorem}
                \begin{proof}
                    The closed and densely defined operator $\mathcal{A}$ and its adjoint $\mathcal{A}^\ast$ are dissipative, i.e.
                      \begin{align*}
                            \angles{A \vect{x},\vect{x}}_X&= -\int_0 ^1 \kappa \dot{\Phi}^2(z)  ~dz\leq 0, \\
                            \angles{A^*{\vect{x}},\vect{x}}_X&= -\int_0 ^1 \kappa \dot{\Phi}^2(z)  ~dz\leq 0,
                      \end{align*} 
                      where we computed $\mathcal{A}^\ast$ in a similar fashion as in Lemma \ref{ch3:lem:skewadjoint_A_EBBT}. Hence, by use of the Lummer-Phillips Theorem \ref{ch3:theorem:Lumer-Phillips}, the operator $\mathcal{A}$ generates a strongly continuous semigroup of contractions and is well-posed.\\
                    From the Sobolev embedding Theorem \cite{Sobolevembedding2006a} we have that $\Dom(\mathcal{A})$ is compact in $\mathcal{X}_\mathcal{H}$ and thus $\mathcal{A}$ is closed. Therefore, the resolvent of $\mathcal{A}$ is compact for all $\lambda$ in the resolvent set \cite{Kato1995}. Using Theorem \cite{Luo1999SSIDSA}, we see that the orbit  $\tilde{\gamma}(\vect{x})$ is precompact, and the limit set is non-empty. It remains to show that the largest invariant subset 
                      $$\mathcal{W}=\{x \mid \dot{\mathcal{V}}(\vect{x})=0\},$$
                      with {$\dot{\Phi}(z)\equiv 0$} contains only the zero vector $\vect{0}$. Therefore, let Assumption \ref{ch3:ass:sys_coef} hold and recall the boundary conditions {$v(0)=w(0)=w_z(0)=\Phi(0)=0$} and compute  the solution to the ode
                        \begin{align}
                            \mathcal{A}\vect{x}(z) = \vect{0}, 
                        \end{align}
                      with use of
                        \begin{align}
                            (C_A - \tfrac{C_{I_0}^2}{C_I}) v_{zz} &= 0 , ~ v(0) = 0, ~v_z(1) = 0 \nonumber\\
                            (C_I  - \tfrac{C_{I_0}^2}{C_A} )w_{zzz} & = 0, ~ w(0)=0 , ~w_z(0) = 0,~ w_{zz}(1) =0 \nonumber\\
                            \tfrac{A^p}{\mu }\Phi_{zz}& = 0,~ \Phi(0) =0 ,~ \Phi_z(1)=0
                        \end{align} 
                        we obtain the solutions $v(z)=0$, $w(0)=0$, and $\Phi(0)=0$, hence we conclude                      $\vect{x}(z)\equiv\vect{0}$ for $\vect{x}\in\Dom(\mathcal{A})$. Therefore, $\vect{0}\in \mathcal{W}$ is the only solution contained in $\mathcal{W}$. Finally, by use of LaSalle's Invariance Principle Theorem \ref{ch3:thm:LaSalleinv}, we have that 
                      \begin{align}
                          \lim_{t\rightarrow\infty} d(T_{cl}(t)\vect{x},\mathcal{W})=d(T_{cl}(t)\vect{x},\vect{0})=0,
                      \end{align}
                      for all $\vect{x}\in \Dom(\mathcal{X}_\mathcal{H})$ and conclude that the closed-loop system is asymptotically stable. 
                  \end{proof}
            \hfill \break
            
            For the controlled piezoelectric composite \eqref{ch3:eq:PDE_currentactuator_TBT_FD} with Timoshenko beam theory, we can show in a similar fashion the asymptotical stabilizability of the closed loop system obtained by control \eqref{ch3:eq:control_choice}. Therefore, define $\theta=w_z-\phi$ and $\vect{x}=\col(v_z,\phi_{z},\theta,\Phi_z,\rho_A \dot{v}-\rho_{I_0}\dot{w}_z,\rho_I \dot{w}_z - \rho_{I_0}\dot{v},\rho_A \dot{w}, \tfrac{A^p }{\varepsilon} \dot{\Phi}+\gamma (A^p v_z- I_0 w_{zz}))$ such that $x\in \mathcal{X}_\mathcal{H}$ and denote the closed-loop system obtained via the control choice \eqref{ch3:eq:control_choice} as follows,
            {\begin{align}\label{ch3:eq:closedloop_sys2}
                \dot{\vect{x}}&=\mathcal{A}_2\vect{x}, t>0 \\
                \vect{x}(0)&=\vect{x}_0 \in \Dom(\mathcal{A}_2) \nonumber
            \end{align}}
            with the closed and densely defined operator
           { \begin{align}\label{ch3:eq:closedloopgen_TBT}
                    &\mathcal{A}_2\vect{x}= \nonumber \\
                    &\begin{pmatrix} \tfrac{\rho_I}{\rho_A\rho_I-\rho_{I_0}^2}x_5' + \tfrac{\rho_{I_0}}{\rho_A\rho_I-\rho_{I_0}^2} x_6'\\ \tfrac{\rho_A}{\rho_A\rho_I-\rho_{I_0}^2} x_6' + \tfrac{\rho_{I_0}}{\rho_A\rho_I-\rho_{I_0}^2}x_5'    \\  \tfrac{1}{\rho_A}x_7' - \tfrac{\rho_{I_0}}{\rho_A\rho_I-\rho_{I_0}^2}x_5 - \tfrac{\rho_{A}}{\rho_A\rho_I-\rho_{I_0}^2} x_6\\ \tfrac{1}{\varepsilon A^p} x_8' - \tfrac{\gamma}{\varepsilon}x_1'+\tfrac{\gamma I_0}{\varepsilon A^p}x_2'  \\
                    (C_A+\tfrac{\gamma^2 A^p}{\varepsilon})x_1'-(C_{I_0}+\tfrac{\gamma^2 I_0}{\varepsilon}))x_2'-\frac{\gamma A_p}{\varepsilon}  x_8' \\
                    (C_I+\tfrac{\gamma^2 I_0^2}{\varepsilon A^p}))x_2'-(C_{I_0}+\tfrac{\gamma^2 I_0}{\varepsilon}))x_1'+\frac{\gamma I_0 }{\varepsilon}x_8'+ \tfrac{1}{2}G_A x_3\\
                    \tfrac{1}{2}G_A x_3' \\
                    \tfrac{A_p}{\varepsilon} x_4' -\kappa (\tfrac{1}{\varepsilon A^p} x_8 - \tfrac{\gamma}{\varepsilon}x_1+\tfrac{\gamma I_0}{\varepsilon A^p}x_2) \end{pmatrix} \nonumber\\
                    &=\begin{pmatrix}
                        \dot{v}_z \\ \dot{\phi}_z \\ \dot{\theta} \\ \dot{\Phi}_z \\
                        C_A v_{zz} - C_{I_0} \phi_{zz} - \gamma A^p \dot{\Phi}_z \\
                        C_I \phi_{zz} - C_{I_0} v_{zz} +\gamma I_0 \dot{\Phi}_z + \tfrac{1}{2}G_A\theta\\
                        \tfrac{1}{2}G_A \theta_z \\ \frac{A^p }{\mu} \Phi_z
                    \end{pmatrix}
            \end{align}}
            on the domain 
          { \begin{align}
                \begin{split}
                        &\Dom(\mathcal{A}_2)=\{ \vect{x} \in \mathcal{X}_\mathcal{H} |  \\
                        & \quad C_A v_z(1)-C_{I_0}\phi_z(1)-\frac{\gamma}{\varepsilon} A^p \dot{\Phi}(1)=0,\\
                        & \quad C_I\phi_z(1)-C_{I_0}v_z(1)+\frac{\gamma}{\varepsilon} I_0 \dot{\Phi}(1)=0,\\ 
                        & \quad \theta(0)=0, ~\tfrac{A^p}{\varepsilon} \Phi_z(1)=0 \}.
               \end{split}
            \end{align}}
            Similarly, to \eqref{ch3:eq:closedloopgen_EBBT}, it is straightforward to show that the operator \eqref{ch3:eq:closedloopgen_TBT} generates a strongly continuous semigroup of contractions, which is well-posed in a similar fashion as is shown for \eqref{ch3:eq:opA_FD_TBT}. Furthermore, we are able to prove the asymptotic stability of the closed-loop system in a similar fashion as is shown for \eqref{ch3:eq:closedloopgen_EBBT}, with the use of the inner-product $\langle \vect{x},\vect{x}\rangle_\mathcal{H}=2\mathcal{H}_{T}(t)$ on $X_\mathcal{H}$ and use Assumption \ref{ch3:ass:sys_coef} for the system parameters.
            \begin{theorem}\label{ch3:thm:Asympstabz_TBT_FD}
                    Let the inequalities in Assumption \ref{ch3:ass:sys_coef} hold and consider the closed-loop system  \eqref{ch3:eq:closedloopgen_TBT}. Furthermore, consider the Lyapunov candidate functions $\mathcal{V}=\mathcal{H}_T(t)$. Then, the closed-loop system \eqref{ch3:eq:closedloop_sys2} is well-posed and asymptotically stable.
                \end{theorem}
                \begin{proof}
                    The closed and densely defined operator $\mathcal{A}_2$ and its adjoint $\mathcal{A}_2^\ast$ are dissipative, i.e.
                      \begin{align*}
                            \angles{A_2 \vect{x},\vect{x}}_X&= -\int_0 ^1 \kappa \dot{\Phi}^2(z)  ~dz\leq 0, \\
                            \angles{A_2^*{\vect{x}},\vect{x}}_X&= -\int_0 ^1 \kappa \dot{\Phi}^2(z)  ~dz\leq 0,
                      \end{align*} 
                      where we computed $\mathcal{A}_2^\ast$ in a similar fashion as in Lemma \ref{ch3:lem:skewadjoint_A_TBT}. Therefore, by use of the Lummer-Phillips Theorem \ref{ch3:theorem:Lumer-Phillips}, the operator $\mathcal{A}_2$ generates a strongly continuous semigroup of contractions $T_{cl2}(t)$ and is well-posed.\\
                      Similarly to Theorem \ref{ch3:thm:Asympstabz_EBBT_FD}, we have from the Sobolev embedding Theorem \cite{Sobolevembedding2006a} that $\Dom(\mathcal{A}_2)$ is compact in $\mathcal{X}_\mathcal{H}$ and thus $\mathcal{A}_2$ is closed. Therefore, the resolvent of $\mathcal{A}_2$ is compact for all $\lambda$ in the resolvent set \cite{Kato1995}. Using Theorem \cite{Luo1999SSIDSA}, we see that the orbit  $\tilde{\gamma}(\vect{x})$ is precompact, and the limit set is non-empty. It remains to show that the largest invariant subset 
                      $$\mathcal{W}=\{x \mid \dot{\mathcal{V}}(\vect{x})=0\},$$
                      with {$\dot{\Phi}(1)\equiv 0$} contains only the zero vector $\vect{0}$. Therefore, let Assumption \ref{ch3:ass:sys_coef} hold and recall the boundary conditions {$v(0)=\phi(0)=\phi_z(0)=\theta(0)=\Phi(0)=0$} and compute  the solution to the ode
                        \begin{align}
                            \mathcal{A}_2\vect{x}(z) = \vect{0}, 
                        \end{align}
                      {with use of
                        \begin{align}
                            \theta_z &= 0 , ~\theta(0)=0 \nonumber \\
                            (C_A - \tfrac{C_{I_0}^2}{C_I}) v_{zz} &= 0 , ~ v(0) = 0, ~v_z(1) = 0 \nonumber\\
                            (C_I  - \tfrac{C_{I_0}^2}{C_A} )\phi_{zz} & = 0, ~ \phi(0)=0 , ~\phi_z(0) = 0\\
                            \tfrac{A^p}{\mu }\Phi_{zz}& = 0,~ \Phi(0) =0 ,~ \Phi_z(1)=0  \nonumber
                        \end{align} 
                        we obtain the solutions $\theta(z)=0$, $v(z)=0$, $\phi(0)=0$, and $\Phi(0)=0$, hence we conclude $\vect{x}(z)\equiv\vect{0}$ for $\vect{x}\in\Dom(\mathcal{A}_2)$. Therefore, $\vect{0}\in \mathcal{W}$ is the only solution contained in $\mathcal{W}$. Finally, by use of LaSalle's Invariance Principle Theorem \ref{ch3:thm:LaSalleinv}, we have that 
                      \begin{align}
                          \lim_{t\rightarrow\infty} d(T_{cl2}(t)\vect{x},\mathcal{W})=d(T_{cl2}(t)\vect{x},\vect{0})=0,
                      \end{align}
                      for all $\vect{x}\in \Dom(\mathcal{X}_\mathcal{H})$ and conclude that the closed-loop system is asymptotically stable. }
                  \end{proof}

            \begin{remark}\label{ch3:rem:notstabz_QS_comp}
                The current-controlled piezoelectric composite and actuator under static $(B_2(=\Phi_z)\downarrow 0)$ and quasi-static electric field assumption $(D_3 (=\varepsilon_{33} \dot{Phi}+\gamma \epsilon_{11} \downarrow 0)$ in \eqref{ch3:eq:PDE_currentactuator_EBBT_FD} and \eqref{ch3:eq:PDE_currentactuator_TBT_FD} are not stabilizable for both beam theories since the electromagnetic dynamics and the current input are decoupled from the bending and stretching equations of the piezoelectric composites or actuators. 
            \end{remark}           
            
            \begin{figure}[h]
         			\centering
            	    	\includegraphics[width=0.95\columnwidth]{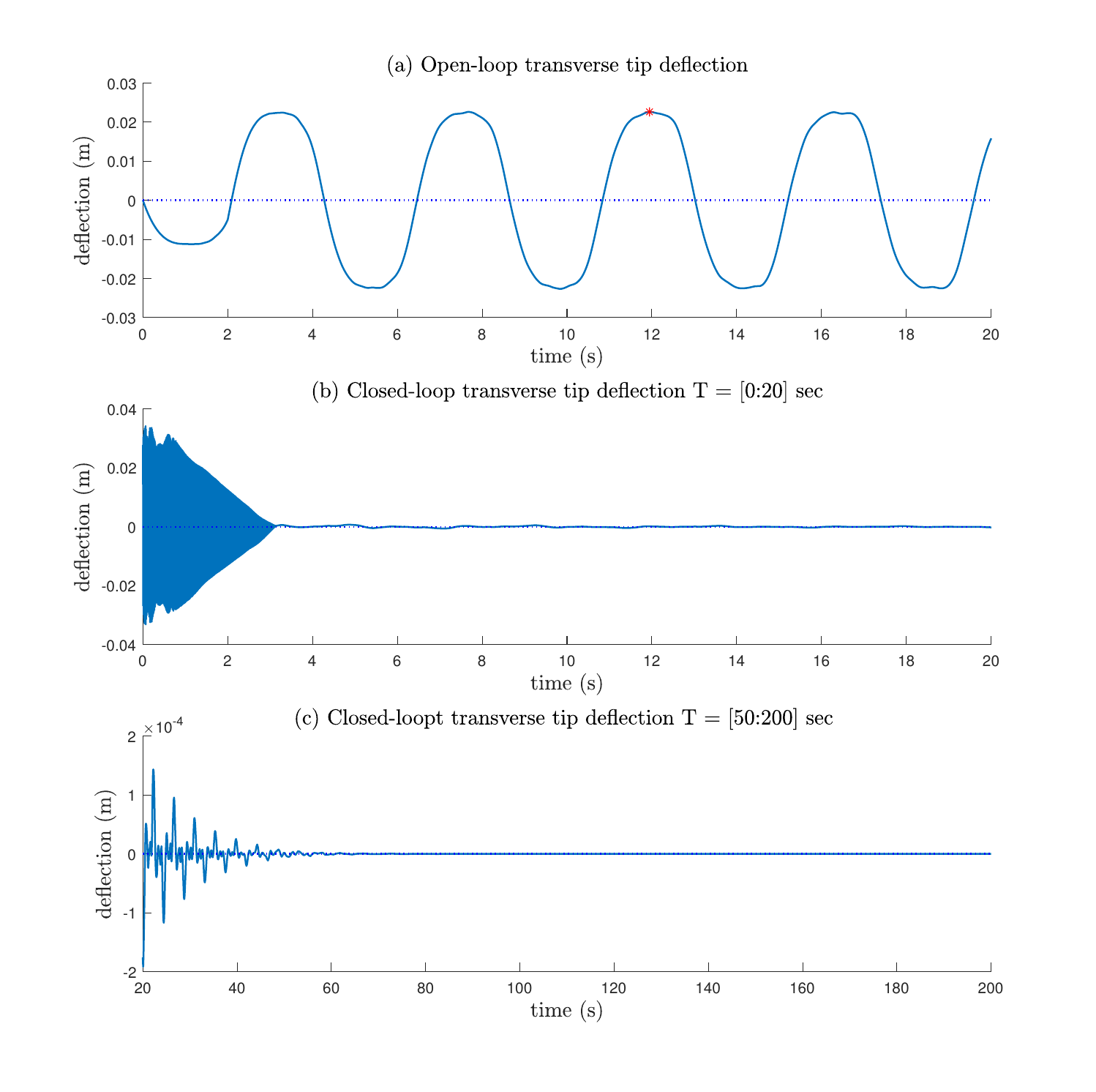}
            				\caption[Open- and closed-loop tip deflection of current-actuated \ac{EBBT} piezoelectric composite]{Open- and closed-loop tip deflection $(w(1))$ of the piezoelectric composite with \ac{EBBT}. In (a), the uncontrolled (open-loop) piezoelectric composite is actuated with 500 A for the first two seconds and reaches a deflection of ~2.2 cm at the tip. After two seconds, the actuation is stopped, and it can be seen that the tip of the composite moves back to zero and continues to vibrate regularly. In (b) and (c), the closed-loop system behaviour for the transverse tip deflection is presented. In (b), the asymptotically stabilizing behaviour is shown, and in (c), the convergence of the deflection to zero is underlined.}
                    \label{ch3:fig:sim_EBBT} 
                \end{figure}
                \begin{figure}[h]
            		\centering
            			\includegraphics[width=0.95\columnwidth]{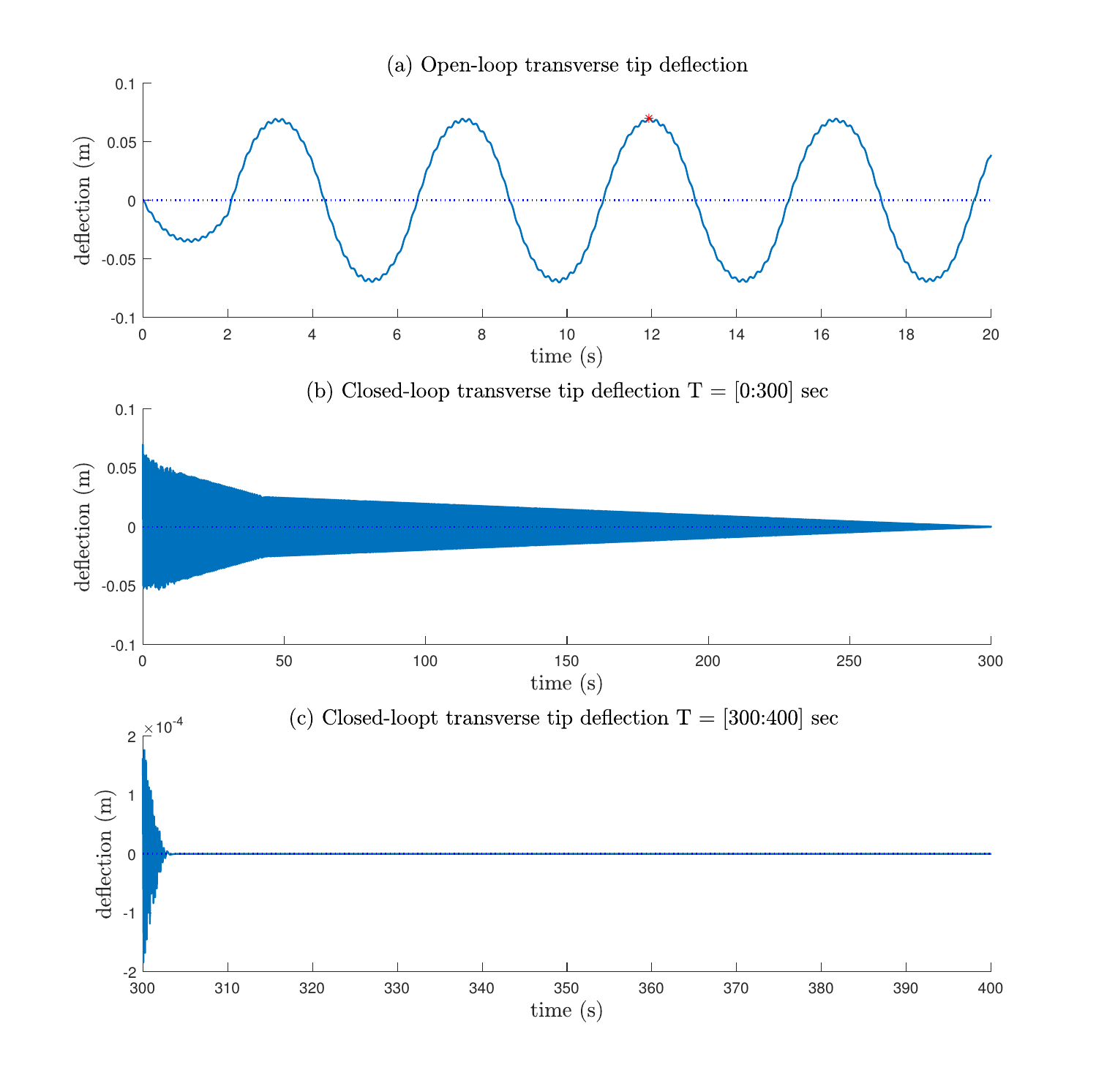}
            				\caption[Open- and closed-loop tip deflection of current-actuated \ac{TBT} piezoelectric composite]{Open- and closed-loop tip deflection $(w(1))$ of the piezoelectric composite with \ac{TBT}. In (a), the uncontrolled (open-loop) piezoelectric composite is actuated with 500 A for the first two seconds, reaching a deflection of ~6,7 cm at the tip. After two seconds, the actuation is stopped, and it can be seen that the tip of the composite moves back to zero and continues to vibrate regularly. In (b) and (c), the closed-loop system behaviour for the transverse tip deflection is presented. In (b), the asymptotically stabilizing behaviour is shown, and in (c), the convergence of the deflection to zero is underlined.}
            			\label{ch3:fig:sim_TBT}
        	    \end{figure}
                
            \section{Simulation results of the asymptotically stabilizing piezoelectric composites}\label{ch3:sec:simulations}
                We include some simulation results for illustrative purposes underlining the asymptotic stabilizability of the piezoelectric composites presented in this work. Therefore, we consider two piezoelectric composites, with the top layer being a  \ac{PZT}-5\footnote{\text{https://support.piezo.com/article/62-material-properties}} piezoelectric layer and the mechanical substrate (with centroidal coordinates) is a steel 304\footnote{https://support.piezo.com/article/62-material-properties\#pack} mechanical layer of the same dimensions. An overview of the system parameters is given in Table \ref{ch3:tab:symbols}. The open-loop and closed-loop simulations, by closing the loop in a standard passive manner \cite{Jacobzwart2012}, are obtained using the structure-preserving discretization method \cite{GoloSchaft2004} with N=20 segments. The time-discretisation is complimented with the variable-step ode23s-solver (build-in Matlab\textsuperscript{\textregistered} solver). Furthermore, to overcome difficulties with the spatial discretization of (piezoelectric) models using a mixed Finite-Element method \cite{GoloSchaft2004}, mentioned in \cite{voss2010port}, we use a trapezoidal spatial integration and time integration to compute the longitudinal $(v(z,t))$ and transversal deflection $w(z,t)$, respectively.
                                
               The transverse tip behaviour of the open-loop and closed-loop current actuated piezoelectric composite with \ac{EBBT}, respectively \eqref{ch3:eq:PDE_currentactuator_EBBT_FD} and \eqref{ch3:eq:closedloopgen_EBBT} are depicted in Fig \ref{ch3:fig:sim_EBBT}. In Fig \ref{ch3:fig:sim_EBBT}(a), the open-loop transverse vibrations are shown, and in Fig \ref{ch3:fig:sim_EBBT}(b) and (c), the closed-loop behaviour with $\kappa=10$ is shown. The starting point $(t=0)$ for the closed-loop system is the top of the third lobe of the open-loop behaviour indicated by the red star in Fig \ref{ch3:fig:sim_EBBT}(a).
                Similarly, for the open-loop and closed-loop current actuated piezoelectric composite with \ac{TBT}, respectively \eqref{ch3:eq:PDE_currentactuator_TBT_FD} and \eqref{ch3:eq:closedloopgen_TBT}, the transverse tip behaviours are depicted in Fig \ref{ch3:fig:sim_TBT}. In Fig \ref{ch3:fig:sim_TBT}(a), the open-loop transverse vibrations are shown, and in Fig \ref{ch3:fig:sim_TBT}(b) and (c), the closed-loop behaviour with $\kappa=10$ is shown. The starting point $(t=0)$ for the closed-loop system is the top of the third lobe of the open-loop behaviour indicated by the red star in Fig \ref{ch3:fig:sim_TBT}(a). Comparing the behaviour of the two different beam theories, we see that the open-loop tip displacement is much larger in the case of using \ac{TBT}. Furthermore, the piezoelectric composite with \ac{EBBT} stabilizes faster than the piezoelectric composite with \ac{TBT}.

             \begin{table}
                    \centering
                    \caption[System parameters]{System parameters}
                    {\begin{tabular}{lll}
                        \toprule
                            Geometry & Description& Value\\
                        \midrule
                            $2g_b$  & Layer width & $0,1 ~m$ \\
                            $h_b-h_a$ & Layer thickness & $0,01 ~m$\\
                        \midrule
                             Piezo parameters     \\
                        \midrule
                            $\rho_p$    &  mass-density     & $7950 ~\frac{\text{kg}}{\text{m}^3}$\\
                            $C_p$         & Stiffness         & $66\times10^9 ~\frac{\text{N}}{\text{m}^2}$\\
                            $\gamma$    &   Coupling coefficient &  $12.54  ~{\frac{\text{C}}{\text{m}^2}}$\\
                            $\varepsilon$     & Impermittivity & $10^{6}~ \frac{m}{F}$ \\
                            $\mu$ & Magnetic permeability & $1.2 \times 10^{-6} ~\frac{H}{m}$  \\
                            \midrule
                            Substrate parameters &  &  \\
                            \midrule
                            $\rho_s$    &  mass-density     & $8000 ~\frac{\text{kg}}{\text{m}^3}$\\
                            $C_s$       & Stiffness         & $193\times10^9  ~\frac{\text{N}}{\text{m}^2}$\\
                        \bottomrule
                        \end{tabular}}
                    \label{ch3:tab:symbols}
                \end{table}

\section{{Discussion and further research}}\label{ch3:sec:discussions}

In this work, we propose two new well-posed current-controlled piezoelectric composite models with a fully dynamic electromagnetic field that consider different beam theories, i.e. the Euler-Bernoulli and Timoshenko beam theories, for the mechanical domain. The novelty lies in using a combined Lagrangian that couples the equations through so-called traditors, which gyrates forces and flows in the system to obtain a system of well-posed dynamical equations and circumvents the use of a gauge condition within the electromagnetic domain. Furthermore, we show that the derived piezoelectric composites derived with the proposed modelling approach are asymptotically stabilizable for certain system parameters.  \blue{Whereas, the fully dynamic electromagnetic systems derived with vector potentials combined with the Coulomb gauge condition are not stabilizable \cite{MenOCDC2014,OzerTAC2019}.} The derived classical passivity-based control law incurs some in-domain electromagnetic damping. It might be interesting to look into this from a material science perspective. Furthermore, it would be interesting to see if it is possible to extend the boundary control Krasovskii passivity control methodologies presented in \cite{Jong2023} for distributive control inputs such as the derived current controlled piezoelectric composite models.

\section*{Acknowlegdement}
{We would like to acknowledge Kirsten A. Morris for her insights and suggestions during the development of this work.}

\bibliographystyle{IEEEtran}

\end{document}